\documentclass[twocolumn, pra,superscriptaddress, amsmath,amssymb, aps]{revtex4}
\usepackage{graphicx}
\usepackage{subfigure}
\DeclareGraphicsExtensions{.pdf,.png,.jpg}
\usepackage{dcolumn}
\usepackage{bm}

\usepackage{physics}
\usepackage{amsfonts, amsmath}
\usepackage{dsfont}
\usepackage{amsthm}
\usepackage{bm}
\usepackage{bbold}
\usepackage{color}
\usepackage{lipsum,babel}
\usepackage[normalem]{ulem}
\usepackage{grffile}
\usepackage{tikz}
\usepackage{braket}
\usepackage{xr-hyper}
\usepackage[colorlinks, linkcolor=blue, anchorcolor=blue, citecolor=blue, urlcolor=blue]{hyperref}
\usepackage{enumitem}
\usepackage{titlesec}

\usetikzlibrary{quantikz}
\DeclareGraphicsExtensions{.pdf,.png,.jpg}

\newcommand{\bfu}{\mathbf{u}}
\newcommand{\bfx}{\mathbf{x}}
\newcommand{\bfv}{\mathbf{v}}

\newcommand{\bfa}{\mathbf{a}}
\newcommand{\bfb}{\mathbf{b}}

\newcommand{\bfZ}{\mathbb{Z}}
\newcommand{\poly}{\mathrm{poly}}

\newtheorem{observation}{Observation}
\newtheorem{theorem}{Theorem}
\newtheorem{corollary}{Corollary}

\newtheorem{definition}{Definition}
\newtheorem{lemma}{Lemma}

\begin{document}
\title{Extending Classically Simulatable Bounds of Clifford Circuits with Nonstabilizer States via Framed Wigner Functions}
\author{Guedong Park}
\affiliation{NextQuantum and Department of Physics and Astronomy, Seoul National University, Seoul, 08826, Korea}

\author{Hyukjoon Kwon}
\email{hjkwon@kias.re.kr}
\affiliation{School of Computational Sciences, Korea Institute for Advanced Study, Seoul, 02455, Korea}
\author{Hyunseok Jeong}
\email{jeongh@snu.ac.kr}
\affiliation{NextQuantum and Department of Physics and Astronomy, Seoul National University, Seoul, 08826, Korea}

\begin{abstract}
The Wigner function formalism has played a pivotal role in examining the non-classical aspects of quantum states and their classical simulatability. Nevertheless, its application in qubit systems faces limitations due to negativity induced by Clifford gates. In this work, we propose a novel classical simulation method for qubit Clifford circuits based on the framed Wigner function, an extended form of the Wigner function with an additional phase degree of freedom. In our framework, Clifford gates do not induce negativity by switching to a suitable frame; thereby, a wide class of nonstabilizer states can be represented positively. By leveraging this technique, we show that some marginal outcomes of Clifford circuits with nonstabilizer state inputs can be efficiently sampled at polynomial time and memory costs. We develop a graph-theoretical approach to identify classically simulatable marginal outcomes and apply it to $\log$-depth random Clifford circuits. We also present the outcome probability estimation scheme using the framed Wigner function and discuss its precision. Our approach opens new avenues for utilizing quasi-probabilities to explore classically simulatable quantum circuits.
\end{abstract}
\maketitle

Applying quantum mechanical principles to computer science has led to the discovery of quantum algorithms~\cite{shor1999, deutsch1992}. However, not every quantum algorithm manifests exponential speedup over classical algorithms as some quantum circuits can be efficiently simulated classically \cite{gottesman1998, nielsen2001, jozsa2013}. The best-known class of such quantum circuits is defined by the Gottesman-Knill theorem \cite{gottesman1998, aronson2004}; circuits consist of an input state in the computational basis and Clifford gates, resulting in a stabilizer state as an output. While the theorem identifies necessary elements for universal quantum computation~\cite{bravyi2005,howard2017}, understanding the hardness of a more restrictive family of experimentally feasible quantum circuits, such as instantaneous quantum polynomial circuits~\cite{bremner10}, quantum random circuits~\cite{bouland2019}, and unitary Clifford circuits with nonstabilizer inputs~\cite{jozsa2013, bu2019, yoganathan2022}, also plays an important role in the near-term demonstration of quantum computational advantage.

Meanwhile, in a physics-oriented direction, the classical simulatability of quantum circuits has been studied based on the Wigner function \cite{wigner1932}, which describes quantum phase space. In quantum optics, Gaussian states are the only pure states with a positive Wigner function~\cite{soto1983}, and Gaussian operations preserve the positivity of the Wigner function. A remarkable similarity can be found in discrete variable quantum phase space with odd-prime dimensions \cite{gross2006}, where stabilizer states and Clifford operations correspond to Gaussian states and Gaussian operations, respectively. This common feature enables the unified construction~\cite{mari2012} of a classical simulation method for both discrete~\cite{veitch2012} and continuous~\cite{veitch2013,keshari2016} variable quantum circuits with positive Wigner functions. Moreover, these approaches open up the possibility to simulate nonstabilizer mixed states with positive Wigner functions~\cite{veitch2012}.

For a qubit system, however, classical simulation based on the Wigner function formalism is no longer applicable due to its exotic features~\cite{kocia2017,raussendorf2017}. A crucial problem arises as Clifford operations can induce negativity in the Wigner function~\cite{kocia2017, Koukoulekidis2022}, which prohibits the classical simulation of most qubit Clifford circuits. Despite considerable efforts to address this problem~\cite{kocia2017, raussendorf2020, zurel2020, zurel2023}, these approaches have their own limitations. For example, the positive representation of the nonstabilizer state may not be fully identified \cite{kocia2017,zurel2020} or phase space points may not be efficiently sampled and tracked~\cite{raussendorf2020,zurel2020,zurel2023}. Thus, constructing a qubit Wigner function formalism that behaves well under Clifford operations and developing classical simulation algorithms based on it remain open problems.

In this Letter, we propose an efficient classical simulation algorithm for unitary Clifford circuits with nonstabilizer inputs based on the qubit Wigner function formalism. To this end, we adopt a generalized notion of the Wigner function parameterized by a family of frame functions~\cite{raussendorf2017, raussendorf2020, Koukoulekidis2022}, while actively utilizing its non-local form. Our key observation is that phase space points transform covariantly under qubit Clifford gates without inducing negativity when the frame is appropriately switched. This leads to a qubit Wigner function formalism consistent with the Gottesman-Knill theorem and also significantly extends the classically simulatable regime of Clifford circuits with nonstabilizer inputs, allowing for efficient sampling of their marginal outcomes. We show that the efficiently simulatable marginal outcomes can be identified by solving the vertex cover problem~\cite{clarkson1983} in graph theory. For $\log$-depth random Clifford circuits, we find that the number of simulatable qubits scales linearly with $n$ for a 1D architecture, while observing a phase transition in its scaling for a completely connected architecture. We also discuss the precision of probability estimation \cite{pashayan2015, Koukoulekidis2022} using our approach.

\textit{Simulating quantum circuits in phase space.---}Suppose a quantum circuit, consisting of initial state $\rho$, unitary evolution $U$, and measurement operators $\{ \Pi_\bfx \}$ with outcomes $\bfx$. One way to classically simulate the outcomes $\bfx$ following the probabilities $p(\bfx) = {\rm Tr} [U \rho U^\dagger \Pi_\bfx]$ is by representing quantum states in phase space~\cite{mari2012,veitch2012,pashayan2015}. This can be done by mapping a quantum state to a distribution in phase space $V$ as $\rho=\sum_{\bfu \in V}  W_{\rho}(\bfu)A(\bfu)$, where $A(\mathbf{u})$ are phase space operators composing an operator basis. The outcome probability is then represented as
 \begin{equation}
 \label{eq:prob_1}
p(\bfx) = \sum_{\bfu \in V} W_{U \rho U^\dagger} (\bfu) P(\bfx|\bfu),
\end{equation}
where $P(\bfx|\bfu) \equiv {\rm Tr} [ A(\bfu) \Pi_{\bfx}]$. The phase space distribution $W_\rho(\bfu)$ is well-normalized but can have negative values, in general, often referred to as \emph{quasi-probability}~\cite{wigner1932}. Nevertheless, some quantum circuits may have both positive $W_{U \rho U^\dagger} (\bfu)$ and $P(\bfx|\bfu)$. In this case, the outcomes $\bfx$ can be classically simulated by sampling the phase space point $\bfu$ from  $W_{U \rho U^\dagger} (\bfu)$, followed by sampling $\bfx$ from the conditional probability distribution $P(\bfx|\bfu)$~\cite{mari2012}. 

The canonical form of phase space for an $n$-qubit system is given by a $2^n \times 2^n$ lattice, $\bfu = (\bfu_x, \bfu_z) = (u_{1x}, u_{2x}, \cdots, u_{nx}, u_{1z}, u_{2z}, \cdots, u_{nz}) \in V_n = \bfZ_2^n \times \bfZ_2^n$, where the phase space operator $A(\mathbf{u})\equiv\frac{1}{2^{n}}\sum_{\mathbf{a}\in V_n} (-1)^{[\mathbf{u},\mathbf{a}]} T_{\mathbf{a}}$ is constructed from the Weyl operator $T_{\bfa} \equiv \bigotimes_{i=1}^n i^{a_{ix} a_{iz}} X^{a_{ix}} Z^{a_{iz}} $. Here, $[\bfu, \bfa] \equiv \bfu_x \cdot \bfa_z + \bfu_z \cdot \bfa_x$ is a symplectic inner product and $X$ and $Z$ are Pauli operators corresponding to shift and boost operations in phase space, respectively. This leads to the \emph{discrete Wigner function} defined as $W_{\rho}(\bfu)\equiv \frac{1}{2^n} {\rm Tr}(\rho A(\bfu))$~\cite{raussendorf2017, pashayan2015, Koukoulekidis2022}, which shares common properties with the continuous-variable Wigner function \cite{wigner1932, raussendorf2017}:
(i) it is real-valued and well-normalized, $\sum_{\bfu \in V_n}W_{\rho}(\bfu)=1$; (ii) it is covariant under translation, $W_{T_{\bfa}\rho T^{\dagger}_{\bfa}}(\bfu)=W_{\rho}(\bfu+\bfa)$; and (iii) its marginals indicate correct outcome probabilities.

On the other hand, the qubit Wigner function has exotic properties compared with the odd prime dimensional or continuous variable Wigner function~\cite{veitch2012,mari2012} that Clifford operations can induce negativity~\cite{raussendorf2017}. For example, a two-qubit state $\ket{00}$ in the computational basis state has a positive Wigner function, but after applying the Hadamard and CNOT gates, the resulting Bell state $\ket\Phi = \frac{\ket{00} + \ket{11}}{\sqrt{2}}$ has a negative Wigner function. This strongly limits the range of classically simulatable quantum circuits using qubit Wigner functions~\cite{kocia2017gk}.

\begin{figure}[t]
    \centering
    \includegraphics[width=1\linewidth]{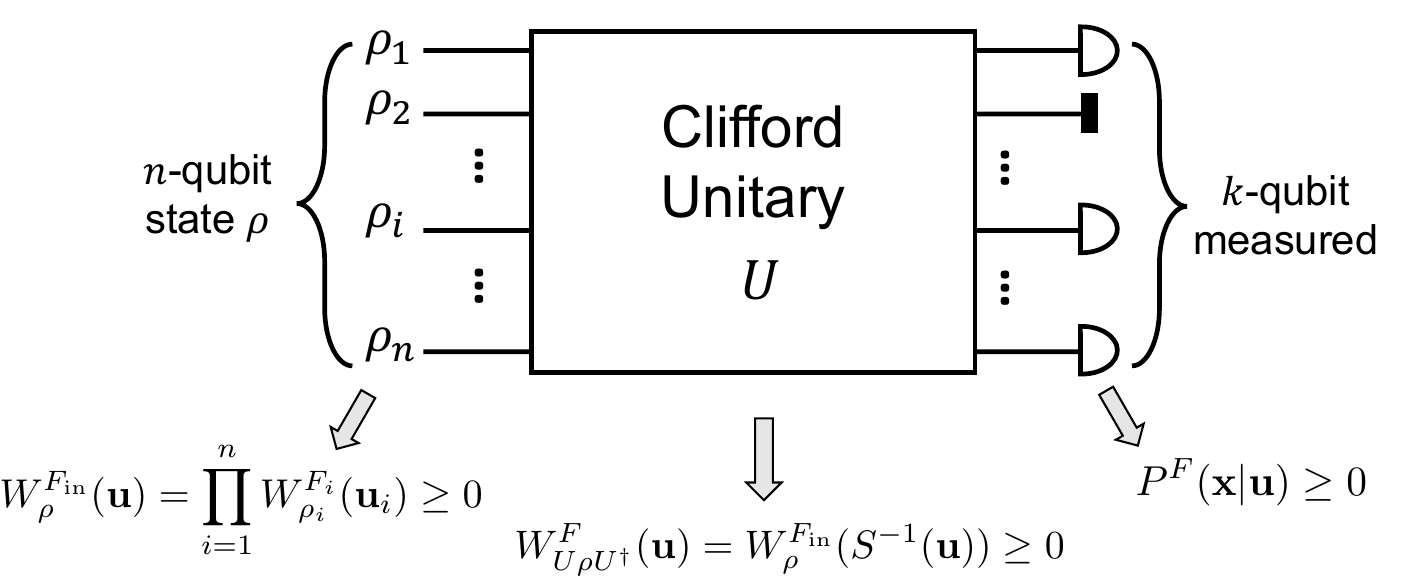}
    \caption{Condition for efficient classical simulation of a Clifford circuit with a nonstabilizer input using the framed Wigner function.}
    \label{fig:fig1}
\end{figure}

\textit{Efficient classical simulation via framed Wigner functions---}
To circumvent this problem, we introduce an additional degree of freedom to the qubit Wigner function~\cite{raussendorf2017,raussendorf2020,Koukoulekidis2022}. We define a \emph{framed Wigner function}
\begin{equation}
        W_{\rho}^{F}(\bfu) \equiv \frac{1}{2^n} {\rm Tr} \left[ \rho A^{F}(\bfu) \right],
\end{equation}
using a \emph{frame function} $F$ that characterizes an additional phase factor of phase space operators,
\begin{equation}
    A^F(\mathbf{u}) \equiv\frac{1}{2^n}\sum_{\mathbf{a}\in V_n} (-1)^{[\mathbf{u},\mathbf{a}]+F(\bfa)}T_{\mathbf{a}}.
\end{equation}
The frame function maps a phase space point to a binary value, i.e., $F(\bfa) \in \{0,1\}$ for $\bfa \in V_n$, and satisfies the condition $F(\mathbf{0})=0$ to ensure that $W_\rho^F$ is well-normalized.
For example, a single qubit has two choices of frames, $F(a_x,a_z) = 0$ (zero frame) or $F(a_x,a_z) = a_x a_x$ (conjugate frame), up to translation symmetry under $T_{(a_x, a_z)}$~\cite{raussendorf2017}.
The framed Wigner function also respects some other properties of the conventional Wigner function, a special case with $F(\bfa) = 0$, such as covariance under translation and having a proper notion of marginals~\cite{Gibbons04, gross2006, raussendorf2017}.

In terms of the framed Wigner function, the probability in Eq.~\eqref{eq:prob_1} can be rewritten as
 \begin{equation}
 \label{eq:frame_p}
p(\bfx) = \sum_{\bfu \in V_n} W_{U \rho U^\dagger}^{F} (\bfu) P^{F}(\bfx| \bfu),
\end{equation}
where $P^F(\bfx|\bfu) \equiv {\rm Tr}[ A^F(\bfu) \Pi_{\bfx}]$. Consequently, the circuit is classically simulatable when both $W^F_{U \rho U^\dagger}(\bfu)$ and $P^F(\bfx|\bfu)$ are positively represented. Our main result states that the additional degree of freedom given by the frame function (or simply frame) leads to a wider class of efficiently simulatable quantum circuits (see also Fig.~\ref{fig:fig1}):
\begin{theorem}\label{thm1}
Suppose an $n$-qubit quantum circuit composed of a product state input $\rho = \bigotimes_{i=1}^n \rho_i$ and a Clifford unitary $U$. If each $\rho_i$ is positively represented in either zero or conjugate frame, the final state $U\rho U^\dagger$ is positively represented in a frame $F$, which can be evaluated from the initial frame with $\mathcal{O}(n^3)$-memory and $\mathcal{O}({\rm poly}(n))$-time costs. From this, one can sample the measurement outcomes of some marginal qubits in the computational basis within $\mathcal{O}(n^2)$-time cost, where these marginal qubits are determined by the frame $F$.
\end{theorem}

We highlight that pure nonstabilizer input states, such as $\ket{A}= \cos (\theta/2) \ket{0} + e^{i (\pi/4)}\sin(\theta/2)\ket{1}$ with $\theta = \cos^{-1}(1/\sqrt{3})$~\cite{qassim2021}, can be positively represented by the framed Wigner function~\cite{raussendorf2017}, as well as their multiple copies such as $\ket{A}^{\otimes n}$. This can be compared with an alternative phase space structure employed in Ref.~\cite{raussendorf2020}, which maintains positivity under Clifford operations and Pauli measurements but encounters difficulties in representing multiple-copy nonstabilizer states positively. We also note that a certain phase space can positively represent any quantum circuit~\cite{zurel2020}, but this does not necessarily imply efficient classical simulation. In contrast, our classical simulation algorithm is proven to be efficient in both time and memory costs.

Compared to the approximate simulation algorithm for average cases~\cite{bu2019}, our algorithm is distinguished by its ability to simulate high Pauli-rank states such as $\ket{A}^{\otimes n}$ and to sample the outcomes following the exact probability distribution. Hence, our approach complements the existing classical simulation algorithms~\cite{raussendorf2017, raussendorf2020, bu2019} by extending a family of classically simulatable quantum circuits, despite their general hardness~\cite{jozsa2013, yoganathan2022}.

The key idea of our algorithm is to defer the negativity induced by Clifford gates, enabling the phase point update until the measurement, summarized as follows:
\begin{observation} [Clifford covariance under frame switching]
\label{obs1}
For any Clifford gate $U$ and input frame $F_{\rm in}$, one can always find a frame $F$, a symplectic transform $S$, and $\bfv \in V_n$
 such that $W_{U \rho U^\dagger}^{F}(S(\bfu) + \bfv) = W_\rho^{F_{\rm in}} (\bfu)$. Consequently, when an input state $\rho$ is positively represented, i.e., $W_\rho^{F_{\rm in}}(\bfu) \geq 0$, the output state can also be positively represented, i.e., $W_{U \rho U^\dagger}^{F}(\bfu) \geq 0$.
\end{observation}
While this stems from the fact that a Clifford gate transforms any Pauli operator to another Pauli operator up to the phase factor~\cite{aronson2004}, a detailed form of $F$ and $S$ can be found in the Supplemental Material~\cite{supple}. For example, the Bell state $\ket\Phi$, which exhibits negativity in the conventional Wigner function, can be positively represented under the frame $F(\bfa) = a_{1x}a_{1z}+a_{1z}a_{2z}a_{1x}+a_{2z}a_{1x}a_{2x} \pmod 2$. This observation can be applied to all stabilizer states and a wide class of nonstabilizer states generated from Clifford circuits. Compared to the case of a fixed framed function with adaptive Pauli measurements, where single-qubit operations are the only operations that preserve the positivity of the Wigner function~\cite{raussendorf2017}, our result shows that positivity is preserved for any unitary, i.e., non-adaptive, Clifford operations by switching the frame. 

We sketch the classical simulation algorithm described in Theorem~\ref{thm1}, starting from the phase space point sampling. The framed Wigner function of a product input state $\rho = \bigotimes_{i=1}^n \rho_i$ is written as $W^{F_{\rm in}}_\rho(\bfu) = \prod_{i=1}^n W^{F_i}_{\rho_i}(u_{ix},u_{iz})$ with $F_{\rm in}(\bfa) = \sum_{i=1}^n F_i(a_{ix}, a_{iz})$, where $F_i$ can be either zero or conjugate frame. Hence, when $W^{F_i}_{\rho_i}$ is positive for every $i=1,\cdots,n$, the phase space point of the final state following the distribution $W^{F}_{U \rho U^{\dagger}}(\bfu)=W^{F_{\rm in}}_{\rho}(S^{-1}(\bfu))$ can be sampled by i) sampling $(u_{ix},u_{iz})$ from each $W^{F_i}_{\rho_i}(u_{ix},u_{iz})$, and then ii) applying the symplectic transform $S$ to $(u_{1x},\cdots, u_{nx}, u_{1z}, \cdots, u_{nz})$ followed by adding $\bfv$ from Observation~\ref{obs1}, which overall takes $\mathcal{O}(n^2)$-time. 

On the other hand, the positive Wigner function does not suffice for efficient classical simulation of the quantum circuit. One issue is that a frame function $F(\bfa)$ may contain the high-degree monomials in $\bfa$, such as $\prod_{i=1}^n a_{ix} a_{iz}$. This gives rise to an exploding number of possible frame choices ($2^{(2^{2n}-2n-1)}$)~\cite{nisan1994}, requiring exponential memory cost. For Clifford circuits with product state inputs, however, we show that the frame functions with polynomials of degree 3, i.e., cubic, are sufficient to run the proposed protocol. In other words, the frame can be written in the form, $F(\bfa) = \sum_{\mu,\nu}c_{(\mu,\nu)} a_{\mu} a_{\nu} + \sum_{\mu,\nu,\omega}c_{(\mu, \nu, \omega)} a_{\mu}a_{\nu}a_{\omega}\pmod{2}$, where $\mu,\nu,\omega\in \left\{1x,\ldots,nx,1z,\ldots,nz\right\}$ and $c_{(\mu, \nu)},c_{(\mu, \nu, \omega)}\in \{0,1\}$. This is because any frame function for each $i$th input qubit is either $0$ or $a_{ix}a_{iz}$, and every frame transformation under a Clifford gate adds up to cubic terms~\cite{supple}. We also note that linear terms can always be absorbed in the translation under $T_\bfa$~\cite{raussendorf2017}. Thus, the possible number of frames is reduced into $2^{\mathcal{O}(n^3)}$, which can be efficiently manipulated with a $\mathcal{O}(n^3)$ bit memory.

For a given phase space point $\bfu$, the outcomes $\bfx$ can be classically sampled under the condition $P^{F}(\bfx|\bfu) \geq 0$. A sufficient condition for efficient sampling of the $n$-qubit computational basis measurement $\Pi_\bfx = \ket{\bfx}\bra{\bfx}$ with $\bfx = (x_1, \cdots, x_n) \in \bfZ_2^n$ can be found as $F(\mathbf{0}_x, \bfa_z) = 0$, which leads to $P^F(\bfx|\bfu) = \bra{\bfx} A^F (\bfu) \ket{\bfx} = (1/2^n)\sum_{ \bfa_z \in \bfZ_2^n } (-1)^{(\bfu_x + \bfx) \cdot \bfa_z +F({\bf0}_x,\bfa_z)} = \delta_{\bfx, \bfu_x}$ by noting that $\bra{\bfx} T_\bfa \ket{\bfx} = \delta_{\bfa_x,\bf0_x} (-1)^{\bfx \cdot \bfa_z}$. This directly leads to the outcome sampling of $\bfx = \bfu_x$ from the given phase space point $\bfu$. Remarkably, any Clifford circuit with a stabilizer input always has positive representation in a frame $F$ obeying this condition so that \textit{all the $n$-qubit outcomes can be efficiently simulated for stabilizer states}, reproducing the result of the Gottesman-Knill theorem~\cite{supple}.

For nonstabilizer inputs, however, it is uncommon to obtain a frame $F$ with the condition $F(\mathbf{0}_x, \bfa_z) = 0$. Nevertheless, an efficient classical simulation of some marginal outcomes is still possible. We note that after tracing out the $j$th qubit, the phase space operator of the remaining qubits is given by ${\rm Tr}_j [A^{F}(\bfu)] = A^{F'}(\bfu')$ with $\bfu' = (\bfu'_x, \bfu'_z) \in V_{n-1}$ with $\bfu'_\lambda = (u_{1\lambda}, \cdots, u_{(j-1)\lambda}, u_{(j + 1)\lambda}, \cdots, u_{n\lambda})$ for $\lambda = x,z$ and $F'(\bfa') = F(\bfa)|_{a_{jx}=a_{jz}=0}$ with $\bfa' \in V_{n-1}$ similarly defined as $\bfu'$. Hence, all the monomials containing $a_{jx}$ or $a_{jz}$ are removed in the reduced frame $F'$. We repeat this step for $(n-k)$ times until the reduced frame of the remaining $k$-qubits satisfies $F'({\bf0}'_{x},\bfa'_z)= 0$. Therefore, the measurement outcomes $\bfx' = (x_1', \cdots, x_k') \in \bfZ_2^k$ on the remaining $k$-marginal qubits $\Pi_{\bfx'} = \ket{\bfx'}_M\bra{\bfx'} \otimes \mathbb{1}_T$ can be efficiently sampled from $\bfu$ as $P^F(\bfx'|\bfu) = {\rm Tr}[A^F(\bfu) \Pi_{\bfx'}]= \langle \bfx' | A^{F'}(\bfu') | \bfx' \rangle=
\delta_{\bfx', \bfu'_x}$, where $M$ and $T$ represent the Hilbert space of the measured and traced-out qubits, respectively.
\begin{figure}[t]
    \centering
    \includegraphics[width=0.95\linewidth]{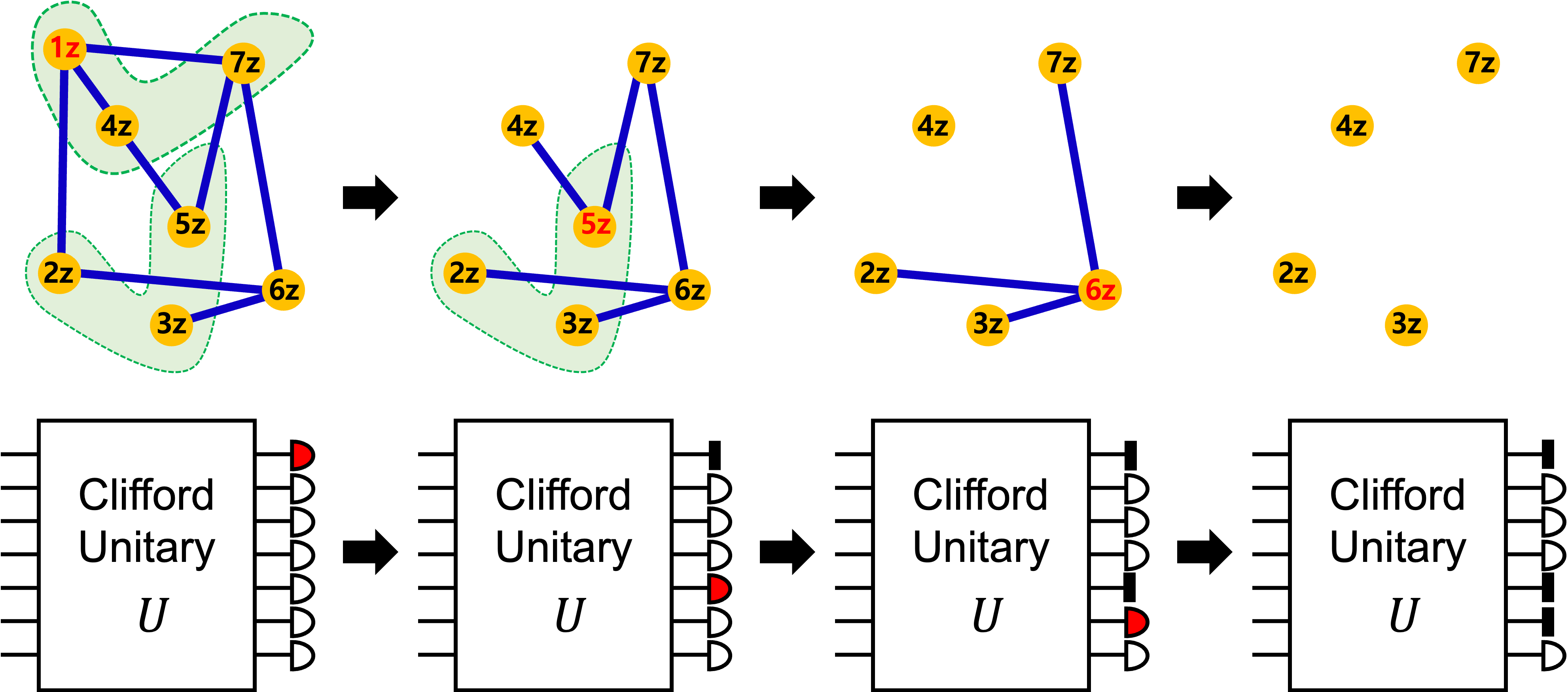}
    \caption{
    Identifying efficiently simulatable qubits by solving the vertex cover problem using a greedy algorithm. Lines and shaded regions represent the hyperedges of ${\cal E}_2$, and ${\cal E}_3$, respectively. A qubit with the largest number of connected hyperedges is traced out at each step until all the hyperedges vanishes. The outcomes of the remaining qubits are efficiently simulatable.
    }
    \label{fig:frame_tracing}
\end{figure}

Furthermore, finding efficiently simulatable marginal qubits can be translated into a graph problem by taking ${\cal V} = \left\{ a_{1z}, \cdots, a_{nz} \right\}$ as vertices and by expressing $F(\mathbf{0}_x,\bfa_z) = \sum_{(i,j) \in {\cal E}_2} a_{iz}a_{jz} + \sum_{(i,j,k)\in {\cal E}_3} a_{iz}a_{jz}a_{kz}\pmod{2}$ as hyperedges ${\cal E}_{2}$ and ${\cal E}_{3}$ connecting two and three vertices, respectively. We note that tracing out the $j$th qubit (taking $a_{jz} = 0$) corresponds to removing the vertex $a_{jz}$ along with all hyperedges connected to it. Hence, removing vertices that fully cover the hyperedges leads to the remaining $k$-qubits satisfying $F'({\bf0}_{x}', {\bfa}_z')=0$. While finding the minimum number of such vertices, the so-called vertex cover problem is an NP-hard problem~\cite{dinur2005}, there exists a sub-optimal algorithm, for example, a \emph{greedy algorithm}~\cite{clarkson1983} running in poly-time (see Fig.~\ref{fig:frame_tracing}).

\textit{Examples.---}We apply the proposed simulation algorithm to two different types of 
random $n$-qubit log-depth Clifford circuits with 1D-neighboring and arbitrary long-range (i.e., completely connected) interactions between two qubits, where gate count is $L=\alpha n \ln (n)$ (see Refs.~\cite{dalzell2022, supple} for more details). We take the input state $\rho = \ket{A}\bra{A}^{\otimes n}$ and adopt the greedy algorithm for solving the vertex cover problem. As the gate count ($\alpha$) increases, the final frame's hypergraph becomes more complex, requiring the removal of more vertices, which results in fewer efficiently simulatable marginal qubits ~\cite{supple}. Figure~\ref{fig:example}(a,c) shows that our simulation method successfully samples the measurement outcomes of these marginal qubits. 

\begin{figure}[t]
    \centering
    \includegraphics[width=0.95\linewidth]{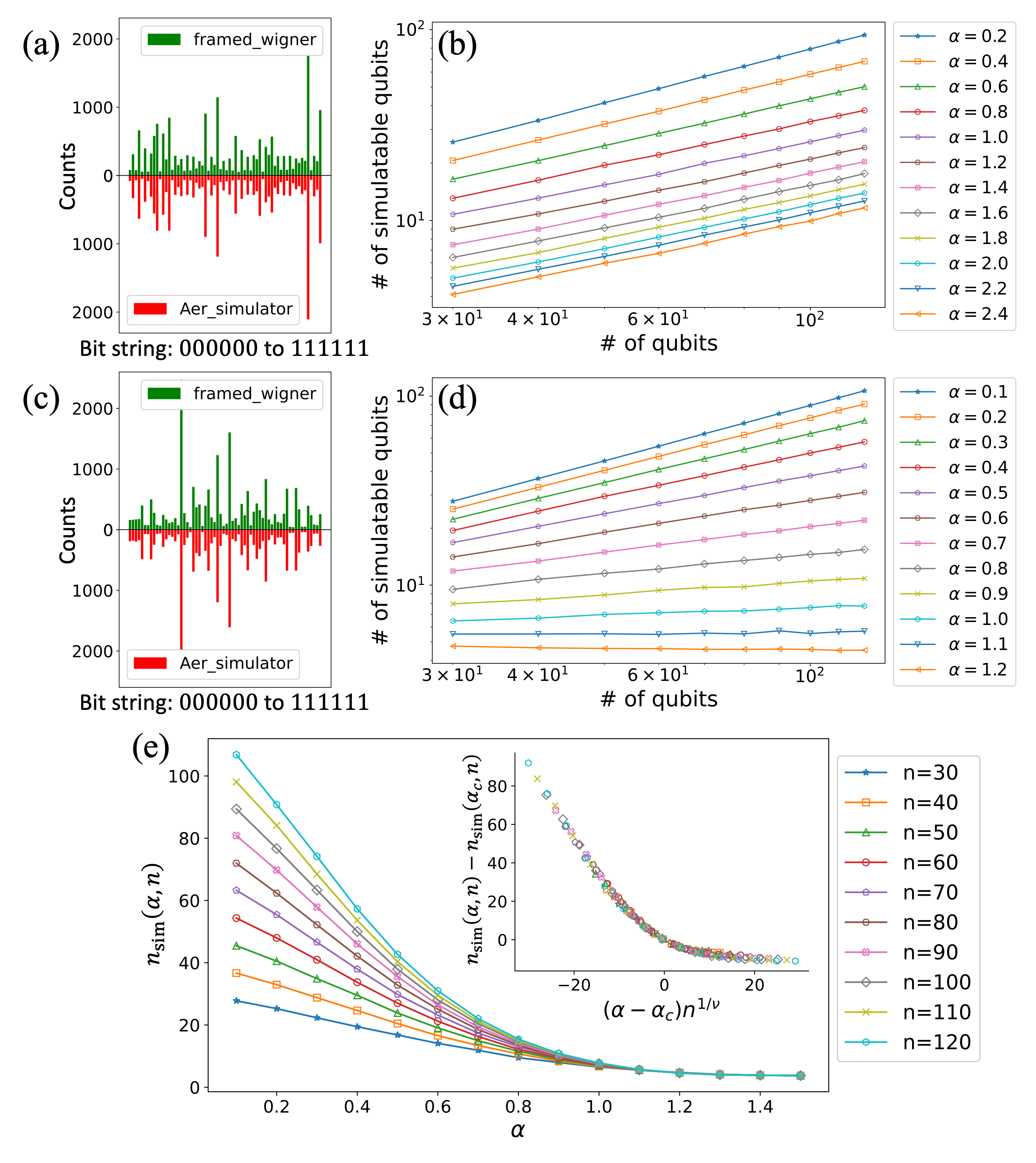}
    \caption{(a), (c): Comparison between sampling results using the framed Wigner function and Qiskit Aer simulator for a randomly chosen $10$-qubit log-depth Clifford circuit with (a) 1D-neighboring and (c) completely connected architecture. For both cases, nonstabilizer input state $\ket{A}\bra{A}^{\otimes 10}$ is taken, and $6$ marginal qubits are selected after solving the vertex cover problem. The $y$-axis shows the number of counts of binary string that simulators sampled by taking a total of $20000$ samples.
    (b), (d): The averaged number of simulatable qubits ($n_{\rm sim}$) by increasing the gate count $L = \alpha (n \ln n)$ of $1000$ random Clifford circuits with (b) 1D-neighboring and (d) completely connected architectures. (e) Scaling behavior of $n_{\rm sim}$ by increasing $\alpha$ and data collapse after finite-size scaling.}
    \label{fig:example}
\end{figure}

Figure~\ref{fig:example}(b) shows that for 1D architecture, the average number of classically simulatable qubits ($n_{\rm sim}$) increases linearly by $n$. This result can be compared to the tensor network simulation, where all the $n$-qubit outcomes of these circuits can be efficiently simulated with $e^{\mathcal{O}(\alpha \ln n)} \sim n^{\mathcal{O}(\alpha)}$ time cost~\cite{napp2022,vidal2003}. In contrast, our algorithm simulates a linear portion of qubits with $\mathcal{O}(\alpha {\rm poly}(n))$ time cost, including both evaluating the final frame $F$ and solving its vertex cover problem~\cite{supple}. Therefore, for the selected marginal outcomes, our approach offers a faster simulation than the tensor network method when the circuit depth $\alpha$ becomes large. From the numerical simulation, we observe remarkable improvement compared to the matrix product state simulator of the IBM Qiskit~\cite{supple}.

For the completely connected architecture, we observe a sharp transition of $n_{\rm sim}$ depending on the gate count. For $\alpha \leq \alpha_c$, $n_{\rm sim}$ scales linearly by increasing the number of qubits. In contrast, if the gate count exceeds a certain value $\alpha \geq \alpha_c$, we observe the sub-linear scaling of the $n_{\rm sim}$ (see Fig.~\ref{fig:example}(d)). From the finite-size scaling analysis by taking a general scaling form $n_{\rm sim}(\alpha, n) - n_{\rm sim}(\alpha_c, n) = f((\alpha-\alpha_c) n^{1/\nu})$ \cite{skinner2019}, we numerically estimate the critical value $\alpha_c \approx 0.81$ and $\nu \approx 1.28$ (see Fig.~\ref{fig:example}(e)). We conjecture that this critical phenomenon is closely related to the anti-concentration properties of the completely connected random circuit \cite{dalzell2022}, as their critical points closely align at $\alpha = 5/6$.

\textit{Born probability estimation.---} 
The framed Wigner function can also be utilized for estimating the Born probability in Eq.~\eqref{eq:frame_p} under less restrictive conditions than outcome sampling. When the final phase space point $\bfu$ can be efficiently sampled from $W_{U \rho U^\dagger}^F(\bfu) \geq 0$, one can take an estimator $\hat{p}_{\bfu}^F(\bfx) = P^F(\bfx|\bfu)$ for each $\bfu$. This leads to an unbiased estimation of the probability $p(\bfx) = \mathbb{E}_{\bfu}[\hat{p}_\bfu^F(\bfx)]$, where $\mathbb{E}_{\bfu} [\cdot]$ denotes averaging over phase space points $\bfu$ from the distribution $W_{U \rho U^\dagger}^F(\bfu)$~\cite{pashayan2015}. The estimator $\hat{p}_\bfu^F(\bfx)$ is not necessarily positive and can be efficiently calculated for every $\bfu$ when $F(\mathbf{0}_x,\bfa_z)$ is quadratic~\cite{montanaro2017, buCMP22, bravyi2019}. One can also apply a procedure similar to the sampling scheme by tracing out $(n-k)$ qubits until $F'({\bf0}_x',\bfu_z')$ becomes quadratic, resulting in an efficient estimation of the $k$-marginal outcome probability. The proposed protocol can be generalized for an arbitrary product state by mixing two different frames without affecting the degree of the reduced frame of the final state~\cite{supple}.

We analyze the precision of the proposed estimator using the mean squared error (MSE), ${\rm Var}_{\rm Wig}(\bfx')\equiv\mathbb{E}_{\bfu}[ | \hat p_{\bfu}^F(\bfx') - p(\bfx')|^2]$~\cite{jerrum1986,blair1985}. The average MSE over all possible outcomes becomes $\overline{\rm Var}_{\rm Wig} \equiv 1/2^k \sum_{\bfx' \in \bfZ_2^k} {\rm Var}_{\rm Wig}(\bfx')= (1-Z_{\rm col}^{(k)})/2^k$, where $Z_{\rm col}^{(k)} \equiv \sum_{\bfx' \in \bfZ^k_2} p(\bfx')^2 \geq 1/2^k$ is the collision probability of the $k$-marginal outcomes ~\cite{bouland2019, dalzell2022}. This improves the previous result using the estimator in terms of the Pauli operator~\cite{pashayan2020} with the average MSE of $\overline{\rm Var}_{\rm Pauli} = (1-1/2^{k}) Z_{\rm col}^{(k)} \geq \overline{\rm Var}_{\rm Wig}$~\cite{supple}.

\textit{Remarks.---}
We have constructed a classical simulation algorithm for a Clifford circuit with nonstabilizer inputs based on the framed Wigner function. Our key observation is that the phase space point can be covariantly transformed under any Clifford gate by switching the frame of the Wigner function, which significantly extends the regime of positively represented states. Our protocol offers a classically efficient sampling of marginal outcomes with efficient time and memory cost, where these outcomes can be identified by solving the vertex cover problem. As examples, we have explored $\log$-depth Clifford circuits and observed that the number of simulatable qubits behaves differently between locally and completely connected circuits.

While our approach of introducing a family of frame functions establishes a clear connection between Clifford operations and positive Wigner functions, even for a qubit system, it also leaves potential extensions and further exploration. A crucial question would be whether the proposed methods can be further extended to classically simulate non-Clifford circuits or adaptive Clifford circuits in the presence of noise.

\textit{Acknowledgements.---}
The authors thank Kyunghyun Baek for helpful discussions. This work was supported by the National Research Foundation of Korea (NRF) grants (NRF-2023R1A2C1006115, RS-2024-00413957, and RS-2024-00438415) and the Institute of Information \& Communications Technology Planning $\&$ Evaluation (IITP) grant (IITP-2021-0-01059 and IITP-2023-2020-0-01606) funded by the Korea government. H.K. is supported by the KIAS Individual Grant No. CG085301 at Korea Institute for Advanced Study.

\clearpage

\widetext
\begin{center}
\textbf{\large Supplemental Material: Extending Classically Simulatable Bounds of Clifford Circuits with Nonstabilizer States via Framed Wigner Functions}
\end{center}

\tableofcontents

\section{Proof of Theorem~1}

In this section, we prove Theorem~1 in the main text. Let us first restate Theorem~1 and clarify each step of the proof.
\begin{theorem}
    Suppose an $n$-qubit quantum circuit composed of a product state input $\rho = \bigotimes_{i=1}^n \rho_i$ and a Clifford unitary $U$. If each $\rho_i$ is positively represented in either zero or conjugate frame, the final state $U\rho U^\dagger$ is positively represented within $\mathcal{O}(n^3)$-memory and $\mathcal{O}({\rm poly}(n))$-time costs. From this, one can sample the measurement outcomes of some marginal qubits in the computational basis within $\mathcal{O}(n^2)$-time cost, where these marginal qubits are determined by the frame $F$.
\end{theorem}

To complete the proof, we separately show the following four main parts of Theorem~1. 
\begin{enumerate}
    \item The Wigner representation of a product state input via zero or conjugate frame. 

    \item Changing the positively representing frame by the Clifford operation and how the Wigner function transforms as a result. This proves Observation~1.

    \item Time and memory complexities of the proposed algorithm.

    \item The capability of classical simulation with a reduced frame after marginalization. 
    
\end{enumerate}

Additionally, we will prove that all stabilizer states can be positively represented under a specific frame so that single-qubit Pauli measurements on all qubits can be efficiently carried out via our method, hence reproducing the result of the Gottesmann-Knill theorem~\cite{aronson2004}. Finally, we will discuss the accessibility of information on sampled outcomes when the final frame has relaxed conditions of its degree.

\subsection{Wigner representation of a product state}
We discuss how a product state can be positively represented under a quadratic frame function. When we consider only a single qubit system, the frame function can have up to second degree. Therefore, there are two frame $F=0$ (\emph{zero frame}) and $F=a_{1x}a_{1z}$ (\emph{conjugate frame}). Now, suppose that for $i\in [n]$, each single qubit state $\rho_i$ is positively represented under $F_i=b_i a_{ix}a_{iz}$ with $b_i = 0$ (zero frame) or $b_i = 1$ (conjugate frame). Also, we define the Pauli operator acting on the $i$-th qubit, $T_{\bfa_i}\equiv i^{a_{ix} a_{iz}}X^{a_{ix}}Z^{a_{iz}}$. Then, we can rewrite $\rho_i$ as 
\begin{align}
    \rho_i=\frac{1}{2}\sum_{\bfu_i,\bfa_i\in \bfZ^2_2}W_{\rho_i}(\bfu_i)(-1)^{[\bfu_i,\bfa_i]+b_ia_{ix}a_{iz}}T_{\bfa_i}.
\end{align}
Consequently, we obtain the Wigner representation of a product state as follows:
\begin{align}
    \rho=\bigotimes_{i}^{n}\rho_i=\frac{1}{2^n}\sum_{\bfu,\bfa\in V_{n}}\left(\prod_{i=1}^{n}W_{\rho_i}(\bfu_i)\right)(-1)^{[\bfu,\bfa]+\sum_{i=1}^{n}b_{i}a_{ix}a_{iz}}T_{\bfa},
\end{align}
where $\bfu=\left(\bfu_1,\bfu_2,\ldots,\bfu_n\right)$, $\bfa=\left(\bfa_1,\bfa_2,\ldots,\bfa_n\right)$, and $T_{\bfa}=\bigotimes_{i=1}^{n}T_{\bfa_i}$. Note that $\rho$ is positively represented under the frame $\sum_{i=1}^{n}b_ia_{ix}a_{iz}$. Hence, the Wigner function of a product state can be represented as a product of single-qubit Wigner functions.

\subsection{Frame changing rules (Observation~1)}\label{sup1:secB}

We explain the frame-changing rules under the Clifford operations and the corresponding transformation of the Wigner function. Suppose an $n$-qubit quantum state $\rho$ is positively represented under the frame $F_{{\rm in}}$. As we discussed in the main text, after applying a Clifford unitary $U$, $U\rho U^{\dagger}$ may have negativity~\cite{raussendorf2017} when the frame is fixed. However, we note that the transformation of $T_\bfa$ under any Clifford unitary $U$ has the following form \cite{aronson2004}:
\begin{equation}\label{sup:fpauli}
    UT_{\bfa}U^\dagger = (-1)^{P(\bfa)}T_{S(\bfa)},
\end{equation}
where symplectic matrices $S$ and phase functions $P(\bfa)$ for each Clifford gate are given in Table~\ref{sup:table}. This leads to 
\begin{align}
    U \rho U^{\dagger}&=\frac{1}{2^n}\sum_{\bfu,\bfa\in V_{n}}W^{F_{{\rm in}}}_\rho (\bfu) \left(-1\right)^{[\bfu,\bfa]+F_{{\rm in}}(\bfa)}UT_{\bfa}U^\dagger \\
    &=\frac{1}{2^n}\sum_{\bfu,\bfa}W^{F_{{\rm in}}}_\rho (\bfu) \left(-1\right)^{[\bfu,\bfa]+F_{{\rm in}}(\bfa)+P(\bfa)}T_{S(\bfa)} \\
    &=\frac{1}{2^n}\sum_{\bfu,\bfa}W^{F_{{\rm in}}}_\rho (\bfu) \left(-1\right)^{[S(\bfu),\bfa]+F_{{\rm in}}(S^{-1}(\bfa))+P(S^{-1}(\bfa))} T_{\bfa} \\
    &=\frac{1}{2^n}\sum_{\bfu,\bfa}W^{F_{{\rm in}}}_\rho (S^{-1}(\bfu)) \left(-1\right)^{[\bfu,\bfa]+F_{{\rm in}}(S^{-1}(\bfa))+P(S^{-1}(\bfa))} T_{\bfa}\\
    &=\frac{1}{2^n}\sum_{\bfu,\bfa}W^{F_{{\rm in}}}_\rho (S^{-1}(\bfu)) \left(-1\right)^{[\bfu,\bfa]+F(\bfa)} T_{\bfa}\\
    &= \sum_{\bfu}W^{F_{{\rm in}}}_\rho (S^{-1}(\bfu)) A^{F}(\bfu).
    \label{sup:change}
\end{align}
\begin{table}[t]
\begin{center}
\begin{tabular}{|c|c|c|}
\hline
\textrm{Clifford gate}&
\textrm{$S$}&
\textrm{$P(\bfa)$}\\
\hline
$CNOT_{i\rightarrow j}$-gate & $\substack{a_{iz}\longleftarrow a_{iz}+a_{jz}\\
    a_{jx}\longleftarrow a_{ix}+a_{jx}}$ & $a_{jz}a_{ix}(a_{iz}+a_{jx}+1)$  \\
\hline
$H_i$-gate & $\substack{a_{ix}\longleftarrow a_{iz}\\
    a_{iz}\longleftarrow a_{ix}}$ & $a_{iz}a_{ix}$\\
\hline
$S_i$-gate & $a_{iz}\longleftarrow a_{ix}+a_{iz}$ & $a_{iz}a_{ix}$ \\
\hline
\end{tabular}\\
\caption{\label{sup:table} Symplectic transformations ($S$) of elementary Clifford gates and corresponding phase functions ($P(\bfa)$). $A\leftarrow B$ means $A$ is transformed to $B$.}
\end{center}
\end{table}
By noting that $U\rho U^\dagger = \sum_{\bfu}W^{F}_{U\rho U^\dagger}(\bfu)A^{F}(\bfu)$ in the new frame $F(\bfa) = F_{{\rm in}}(S^{-1}\bfa)+P(S^{-1}\bfa)$, we obtain $W^{F}_{U\rho U^{\dagger}}(\bfu)=W^{F_{{\rm in}}}_{\rho}(S^{-1}(\bfu))$. This proves Observation~1.

We further note that all linear terms within the frame $F(\bfa)$ can be replaced by a translation in the Wigner function. By noting that linear terms can be expressed as $[\bfv,\bfa]$ for some $\bfv\in V_{n}$, we obtain  
\begin{align}
     W^F_\rho(\bfu) &= \frac{1}{4^n} \sum_{\bfa \in V_n} (-1)^{[\bfu, \bfa] + F(\bfa)}{\rm Tr}[\rho T_\bfa]\\
     &= \frac{1}{4^n}  \sum_{\bfa \in V_n} (-1)^{[\bfu + \bfv, \bfa] + (F(\bfa) + [\bfv,\bfa])}{\rm Tr}[\rho T_\bfa]\\
     &= \frac{1}{4^n}  \sum_{\bfa \in V_n} (-1)^{[\bfu + \bfv, \bfa] + F'(\bfa)}{\rm Tr}[\rho T_\bfa]\\
     &= \frac{1}{2^n}  {\rm Tr} [\rho A^{F'}(\bfu + \bfv)]\\
     &= W^{F'}_\rho(\bfu + \bfv),
     \label{sup:linear_ig}
\end{align}
where $F'(\bfa) = F(\bfa) + [\bfv, \bfa]$.
In other words, if a frame $F$ has a linear term, then its Wigner function of $\sigma$ follows by translating the arguments of the original Wigner function.


\subsection{Time and memory complexities of frame changing}
\label{supp:sec_complexity}
We discuss the time and memory complexities of the proposed simulation algorithm. We first discuss the memory cost for storing the frame functions. Suppose we have a product state input $\rho$ which is positively represented under the frame $F_{{\rm in}}=\sum_{i=1}^{n}b_ia_{ix}a_{iz}$ with $b_i\in \left\{0,1\right\}$. We recall that the final frame after the Clifford operation has the following transformation rule, $F(\bfa)=F_{{\rm in}}(S^{-1}\bfa)+P(S^{-1}\bfa)$. Since all Clifford operations can be generated by gates in the set $\{ CNOT, H, S\}$ shown in Table.~\ref{sup:table} with phase functions of degree up to 3 and linear transformation $S^{-1}$ do not raise the degree of frame function, the resulting frame must have up to the third degree. In light of those facts, we can formalize the final frame as the following cubic binary valued polynomial,
\begin{align}\label{eq:frame_function}
    F(\bfa) = \sum_{\mu}c_{(\mu)}a_{\mu}+\sum_{\mu,\nu}c_{(\mu, \nu)} a_{\mu} a_{\nu} + \sum_{\mu,\nu,\omega}c_{(\mu,\nu,\omega)} a_{\mu}a_{\nu}a_{\omega}\pmod{2}, 
\end{align}
where $\mu,\nu,\omega\in\left\{1x,\ldots,nx,1z,\ldots,nz\right\}$. To store this frame information, we need $\mathcal{O}(n^3)$-memory to contain all coefficients of possible monomials.

We then discuss the time and memory costs for updating the frames for each Clifford gate acting on at most two qubits. With $\mathcal{O}(1)$-time and memory complexity, we can find the symplectic matrix $S$ and the phase function $P$ for this gate. Given a frame $F$ to be changed, we collect all monomials having variables in $\left\{a_{ix},a_{iz},a_{jx},a_{jz}\right\}$.  As we collect each monomial, we linearly transform it via $S^{-1}$ and obtain the set of newly generated monomials. This takes constant time. Since we need at most $\mathcal{O}(n^3)$ memory to record arbitrary cubic frame functions and the above steps do not raise the degree of the output frame, newly generated monomials can be recorded in another $\mathcal{O}(n^3)$-memory. Moreover, we can obtain the set of generated monomials by putting or deleting the generated monomial in the new memory, as we transform each monomial in $F$ with constant time. After we get all generated monomials, the summation between $F$ and the generated set, also with the monomials of $P(S^{-1}(\bfa))$ can be done in $\mathcal{O}(n^3)$. Therefore, the complexity for frame changing by each Clifford gate is $\mathcal{O}(n^3)$ and the total time complexity is at most $\mathcal{O}({\rm poly}(n))$ given that we have ${\rm poly}(n)$-number of $2$-qubit Clifford gates.

After getting through all the Clifford gates, by the last argument of Section~\ref{sup1:secB}, we can always choose $\bfv_{F}\in V_{n}$ for the final frame without linear terms by applying an additional translation to the Wigner function to be $W^{F_{\rm in}}_{\rho}(S^{-1}(\bfu+\bfv_F))$, i.e., $W^{F}_{U\rho U^{\dag}}(S(\bfu)+\bfv_F)=W^{F_{{\rm in}}}_{\rho}(\bfu)$. Also, this translation does not affect the scale of total time complexity. From now on, for convenience, we will always assume that the final frame has no linear terms.

\subsection{Weak simulation}\label{supsec1:C}
Now, we discuss sufficient conditions of weak simulation via the framed Wigner function. Suppose the frame $F$ satisfies $F(\mathbf{0}_x,\bfa_z)=0$. We then note that $P^{F}(\bfx|\bfu)$ with $\bfx\in \bfZ^n_2$ can be expressed as
\begin{align}
    P^{F}(\bfx|\bfu)&=\braket{\bfx|A^{F}(\bfu)|\bfx}\\
    &=\frac{1}{2^{n}}\sum_{\bfa\in V_{n}}(-1)^{[\bfu,\bfa]+F(\bfa)}\braket{\bfx|T_{\bfa}|\bfx}\\
    &=\frac{1}{2^{n}}\sum_{\bfa\in V_{n}}(-1)^{[\bfu,\bfa]+F(\bfa)}\delta_{\bfa_x,0}(-1)^{\bfa_z \cdot \bfx} \\
    &=\frac{1}{2^{n}}\sum_{\bfa_z\in \bfZ^{n}_2}(-1)^{\bfa_z\cdot(\bfx+\bfu_x)}\\
    &=\delta_{\bfx,\bfu_x}.
\end{align}
Next, we express the Born probability $p(\bfx)={\rm Tr}(U\rho U^{\dagger}\ket{\bfx}\bra{\bfx})$ in terms of the framed Wigner function. When the output state $U\rho U^\dagger$ is positively represented as $W^{F}_{U \rho U^{\dag}}(\bfu)=W^{F_{\rm in}}_{\rho}(S^{-1}(\bfu+\bfv_F))$, we observe that
\begin{align}
    {\rm Tr}(U\rho U^{\dagger}\ket{\bfx}\bra{\bfx})&=\sum_{\bfu,\bfa\in V_n}W^{F_{\rm in}}_{\rho}(S^{-1}(\bfu+\bfv_F))\braket{\bfx|A^{F}(\bfu)|\bfx}\\
    &= \sum_{\bfu,\bfa}W^{F_{\rm in}}_{\rho}(S^{-1}(\bfu+\bfv_F))\delta_{\bfx,\bfu_x}\\
    &=\sum_{\bfu_z\in \bfZ^n_2,\bfa\in V_n}W^{F_{in}}_{\rho}(S^{-1}((\bfx+(\bfv_{F})_x,\bfu_z+(\bfv_F)_z)).
\end{align}
By the above result, we can make a weak simulation scheme, which samples outcomes $\bfx$ following the probability distribution $p(\bfx)$, as follows:
\begin{enumerate}
    \item Sample the phase point $\bfu\in V_{n}$ from $W^{F_{\rm in}}_{\rho}(\bfu)$ with the initial frame $F_{\rm in}$.
    \item Update $\bfu\leftarrow S(\bfu)$
    \item Update $\bfu\leftarrow \bfu+\bfv_F$
    \item Output $\bfu_x$.
\end{enumerate}

Even if the resulting frame does not satisfy the condition $F(\mathbf{0}_x,\bfa_z)=0$, we could find marginal measurements in which the reduced frame satisfies that condition. We explain this in more detail here. Without losing the generality, assume that we measure only the first to $k(\le n)$-th qubits and then trace out the others. The marginal measurement probability to obtain $\bfx'\in \bfZ^k_2$ then becomes
\begin{align}
    p(\bfx') &= \sum_{\bfx''\in \bfZ^{n-k}_2}{\rm Tr}{U\rho U^{\dagger}\ket{\bfx'\oplus\bfx''}\bra{\bfx'\oplus \bfx''}}\\
    &=\sum_{\bfx''\in \bfZ^{n-k}_2}\sum_{\bfu,\bfa\in V_n}W^{F_{\rm in}}_{\rho}(S^{-1}(\bfu+\bfv))\braket{\bfx' \oplus \bfx''|A^{F}(\bfu)|\bfx' \oplus \bfx''}
    \\&=\frac{1}{2^{k}}\sum_{\bfu}W^{F_{\rm in}}_{\rho}(S^{-1}(\bfu+\bfv))\sum_{\bfa_z'\in \bfZ^{k}_2}(-1)^{\bfu_x\cdot(\bfa_z'\oplus\mathbf{0}'')+\bfa_z'\cdot \bfx'+F(\mathbf{0},\mathbf{\bfa'_z}\oplus \mathbf{0}'')}.
\end{align}
Hence, if the $F'(\mathbf{0}_x,\bfa'_z)=F(\mathbf{0}_x,\mathbf{\bfa'_z}\oplus \mathbf{0}'') = 0$, the last equation becomes $\frac{1}{2^{k}}\sum_{\bfu}W^{F_{\rm in}}_{\rho}(S^{-1}(\bfu+\bfv_F))\delta_{\bfx',\bfu_x}$. Then, we can efficiently measure the outcome $\bfx'$ by similar steps to the main algorithm. We just need to replace $\bfu_x$ with $(u_{1x},\ldots,u_{kx},0,\ldots,0)$. In the same way, if the resulting frame after marginalizing arbitrary $k$-qubits satisfies $F'(\mathbf{0}'_x,\bfa'_z)=0$, we can efficiently measure the marginal outcome. We note that after marginalizing $n-1$ qubits, the reduced frame must become zero.

Now, let us recollect the results obtained through this section to encapsulate the proof. We assume that a input state is a product state $\rho=\bigotimes _{i=1}^{n}\rho_i$ and each $\rho_i$ is positively represented under a single qubit frame $F_i=0$ or $F_i=a_{ix}a_{iz}$. Then $\rho$ can be positively represented under $F_{\rm in}=\sum_{i=1}^{n} b_i a_{ix}a_{iz}\;(\forall b_i\in\left\{0,1\right\})$. As we discussed in the previous section, the time complexity of frame changing is up to ${\rm poly}(n)$ as well as obtaining $S$~\cite{aronson2004} with at most $\mathcal{O}(n^3)$ memory cost. Marginalizing the resulting frame until it becomes linear takes at most $\mathcal{O}(n^4)$-time. This is because checking if a given frame has non-linear terms takes $\mathcal{O}(n^3)$-time, and the reduced frame must be zero after the marginalization of $n-1$ qubits. Also, updating $\bfu$ to $\bfu'$ is a simple matrix multiplication which takes $\mathcal{O}(n^2)$ time. This completes the proof of Theorem~1.
In Section~\ref{sup:sec_hypergraph}, we provide a systematic program to solve this problem by solving the graph theoretical problem.

\subsection{Efficient simulation of Pauli measurements to stabilizer states}
If the depth of the circuit becomes high, the resulting frame may contain many quadratic and cubic terms. Hence, only a small number of qubits might be efficiently simulatable. However, we show that if the input is a stabilizer state, any non-adaptive Clifford circuit can be transformed to the \emph{CH-form} \cite{bravyi2019}, which satisfies the efficient sampling condition. This can be shown by the following Lemma:
\begin{lemma}\label{sup:lem1}[Bruhat Decomposition \cite{dmitri2018, bravyi2021}]
    (i) An arbitrary $n$-qubit Clifford circuit can be rewritten by layers ${\rm hF'}-{\rm SW}-{\rm H}-{\rm hF}$, where ${\rm H}$ is a layer of Hadamard gates, ${\rm SW}$ is a layer of SWAP gates (hence of CNOT gates) and ${\rm hF},{\rm hF'}$ are CNOT-CZ-S-Pauli layered circuits. This decomposition can be done with $poly(n)$-time.

    (ii) Starting from the zero frame $F_{\rm in}=0$, the changed frame $F$ after ${\rm Pauli}-{\rm H}-{\rm hF}$ section satisfies $F({\bf0}_x,\bfa_z)=0$.
\end{lemma}
\begin{proof}
    The proof of (i) can be found in Refs.~\cite{dmitri2018,bravyi2021}. We show (ii) by noting that the phase functions (see Table.~\ref{sup:table}) of $H,S,CNOT$ gates do not have quadratic or cubic monomials with only $a_{iz}$ terms, and the symplectic transforms of $S$ and $CNOT$ gates do not change $a_{ix}$ to $a_{jz}$. Furthermore, we note that the phase function of the Pauli operator is linear, and the symplectic operation is identity. Hence, the final frame $F(\mathbf{0},\bfa_z)$ must be linear, and those linear terms can be converted into translation following the last arguments of Section~\ref{supp:sec_complexity}.
\end{proof}
We can represent stabilizer states as zero state input rotated by a Clifford operation. Also, by Lemma~1, this operation can be transformed to ${\rm hF'}-{\rm SW}-{\rm H}-{\rm hF}$ form. However, the first CNOT-CZ-S part of $hF'$ section and ${\rm SW}$ section do not affect to zero state. Hence, we have a zero state followed by ${\rm Pauli}-{\rm H}-{\rm hF}$ sections. Note the zero states can be positively represented under zero frame, and the resulting frame satisfies $F(\mathbf{0}_x,\bfa_z) = 0$. Therefore, all the $n$-qubit measurement outcomes $\bfx$ can be efficiently simulated via the framed Wigner function. 

\section{Finding efficiently simulatable qubits by solving the vertex cover problem} \label{sup:sec_hypergraph}
In the main text, we discuss that the $k$-marginal qubits are efficiently simulatable when tracing out the $(n-k)$ qubits until the reduced frame meets the condition $F'({\bf0}_x,\bfa'_z) = 0$. In this section, we discuss how this can be translated into a graph theoretical problem, known as \textit{the vertex cover problem} in more detail.
\subsection{Basic notation of graph theory} \label{sup3:secA}
Here, we introduce a formal definition of hypergraph and the vertex cover problem. 
\begin{definition}[Hypergraph]\label{sec3:def1}
    (i) Let $V$ be a non-empty set.
    The \emph{hypergraph} $G(V, E)$ is defined by a tuple of the vertex set $V$ and the edge set $E$, a set of hyperedges $e \subset V$.
    

    Let $G(V,E)$ be a hypergraph.

    (ii) We say $G$ is of $k$th-degree ($k\in \mathbb{N}$) if all edges in $G$ contain at most $k$ vertices.
     
    (iii) If all elements of $E$ have $(k\in\mathbb{N})$-number of elements in $V$, then we call $G(V,E)$ as a \emph{$k$-uniform hypergraph} or simply a \emph{$k$-graph}. Also, a 2-graph is just called a \emph{graph}.
\end{definition}

Next, we define several properties of a hypergraph.
\begin{definition}
    Let $G(V,E)$ be a hypergraph.
    
    i) $V'\subset V$ is a \emph{vertex cover} if all edges in $E$ contain some elements in $V'$. $A\subset V$ is the \emph{minimal vertex cover} if every vertex cover $V'\subset V$ satisfies $|V'|\ge |A|$.
    
    ii) $S'\subset V$ is an \emph{independent set} if any two elements in $S'$ are not contained in same edge in $E$. $B\subset V$ is the \emph{maximal independent set} if every independent set $S'\subset V$ satisfies $|S'|\le |B|$.
    
    iii) We denote $\nu(G)\equiv|A|$ as the size of minimal vertex cover and $\triangle(G)\equiv|B|$ as the size of maximal independent set. We note that minimal(maximal resp.) vertex cover(independent set) could not be unique, but $\nu(G)$ and $\triangle(G)$ are unique.

    iv) For $v\in V$, we define the \emph{degree} of $v$, $d(v)$ as the number of edges containing $v$.

    v) \emph{Degree} of the graph, $d(G)$ is $\max_{v\in V}\left\{d(v)\right\}$. 
\end{definition}

The \emph{vertex cover problem} is to find the minimal vertex cover of a given hypergraph. Now, we obtain the following result. 
\begin{corollary}\label{supp:graphlem} 
    For any hypergraph $G(V,E)$, $\nu(G)+\triangle(G)\le |V|$.
\end{corollary}
\begin{proof}
    Consider a maximal independent set $S\subset V\;(|S|=\triangle(G))$ and suppose $|V\backslash S|=|V|-\triangle(G)<\nu(G)$. Then $V\backslash S$ must not be the vertex cover. Hence, there exists $e\in E$ such that $e$ does not have any elements in $V\backslash S$. Since $e$ is non-zero, without loss of generality, say $e=\left\{v^e_1,v^e_2,\ldots,v^e_k\right\}\subset S\;(k\ge 2)$. Since $S$ is independent, $|e|=1$, which contradicts that $k\ge2$.
    In conclusion, $|V\backslash S|=|V|-\triangle(G)\ge \nu(G)\Rightarrow \nu(G)+\triangle(G)\le|V|$.
\end{proof}
\noindent If the graph representation is $2$-graph (resulting frame is quadratic), then it is known that $\triangle{G}=|V|-\nu(G)$.

The vertex cover problem is an NP-hard problem \cite{halperin2002}, but several efficient and approximative algorithms are valid \cite{guruswami2020}. These algorithms may obtain vertex covers such that the size is larger than the minimal cover but is within a reasonable scale factor. The typical example we use throughout this paper is a \emph{greedy algorithm}. The detailed procedure of the algorithm is as follows. Suppose we have a hypergraph $G(V,E)$.
\begin{enumerate}
    \item For each vertex $v\in V$ of $G$, count $d(v)$ (\emph{degree of $v$}), the number of edges containing $v$.

    \item Take $v'={\rm argmax}_{v\in V}\left\{d(v)\right\}$.

    \item Remove $v'$ from $G$ and also remove all edges containing $v'$. 

    \item Repeat the above sequences with at most $\left|V\right|$ times until no edges are left. 
\end{enumerate}
Step~1 takes at most $\mathcal{O}(|E|)$-time, and Step~2 takes $\mathcal{O}(|V|)$-time. Step~3 takes at most $\mathcal{O}(|E|)$-time, which is the time complexity of set subtraction. Since we repeat these steps at most $|V|$-time, the total time complexity of the greedy algorithm is at most $\mathcal{O}(|V||E|)$. 

\subsection{Graph representation of the frame function and marginalization}

Now, we show how to connect finding simulatable marginal qubits and the vertex cover problem of the hypergraph.
We recall the brief explanation of the graph representation of the frame function in the main text. We shall call this a \emph{frame graph}.

\begin{definition}
    Suppose we have a frame function $F$ of at most $3$rd-degree such that $F(\mathbf{0}_x,\bfa_z) = \sum_{(i,j) \in {\cal E}_2} a_{iz}a_{jz} + \sum_{(i,j,k)\in {\cal E}_3} a_{iz}a_{jz}a_{kz}\pmod{2}$ for proper index set ${\cal E}_2\subset [n]^{\otimes 2},{\cal E}_3\subset [n]^{\otimes 3}$. \emph{Frame graph} $G_{F}$ of $F$ is a hypergraph $G(V,E)$ where $V$ is a set of vertex $\left\{a_{1z},a_{2z},\ldots,a_{nz}\right\}$ and $E= \mathcal{E}_2 \bigcup \mathcal{E}_3$.
\end{definition}

From the Section~\ref{supsec1:C}, we know that whenever we marginalize the $i$-th qubit, we substitute $a_{iz}=0$ to the representing frame $F'(\mathbf{0}_x,\bfa_z)$. In other words, all monomials containing $a_{iz}$ vanish. In a graphical notation, this means that starting from the frame graph of $F'(\mathbf{0},\bfa_z)$, we eliminate both the $i$-th vertex and all connecting hyperedges. Therefore, tracing the qubits to make the resulting frame function zero is equivalent to finding vertex cover: eliminating vertices and connecting edges until all edges vanish. We can do this using the greedy algorithm and note that $|V|$ is at most $n$ and $|S|$ is within $\mathcal{O}(n^3)$. From the arguments in Section~\ref{sup3:secA}, the total time complexity is at most $\mathcal{O}(n^4)$.

Solving the vertex cover problem of the final frame enables us to search simulatable qubits from a given highly entangled circuit that is hard to pick by hand. Also, we can use various modern approximation techniques to find more qubits over the greedy algorithm \cite{halperin2002}.

We can find the largest number of simulatable qubits if we find the minimal vertex cover. Also, we note that the independent set of the frame graph is also a set of simulatable qubits. However, Corollary~\ref{supp:graphlem} says that finding minimum vertex cover and tracing out the qubits corresponding to the vertices produces a larger set of simulatable qubits than finding the maximal independence set.

\section{Simulations on log-depth Clifford circuits}

\subsection{Circuit architecture}
We first briefly explain some definitions regarding two-different $n$-qubit Clifford circuit architectures \cite{dalzell2022}. One is the \emph{1D} architecture, which consists of alternating layers with $2$-qubit random Clifford gates between neighboring qubits (see Fig.~\ref{sup:fig1}(a)). We also regard the Clifford gates connecting the first and the last qubits as a neighboring gate. The other is the \emph{completely connected} architecture, where we put a random single qubit gate to each qubit, and then random Clifford gates are applied between randomly chosen two qubits, regardless of their locations. \emph{Gate count} is the total number of 2-qubit random Clifford gates. In both architectures, random Clifford gates are uniformly chosen~\cite{dalzell2022}, and we enact each gate until the gate count reaches the designated value. The \emph{depth} of the circuit is defined as the minimum value of the number of layers in which a set of 2-qubit gates can be applied in parallel (also, see Ref.~\cite{dalzell2022} for more detail). Therefore, for 1D cases of the gate count $L= \alpha n \ln (n)$, the depth is $\alpha  \mathcal{O}( \ln (n))$. In completely connected cases, given that $L= \alpha n \ln (n)$, the depth does not exceed $\alpha  \mathcal{O}( (\ln (n))^2)$~\cite{brown2015, dalzell2022} with $1-\mathcal{O}(\frac{1}{\poly(n)})$ probability. Furthermore, the authors in Ref.~\cite{dalzell2022} showed that in both cases, there exists sufficient and necessary scaling of gate count $\mathcal{O}(n\ln(n))$  such that outcome probability distribution of random circuit sampling (including random unitary gates) satisfies the anti-concentration condition~\cite{bouland2019}.

\begin{figure*}[b]
\centering
    \includegraphics[width=1\linewidth]{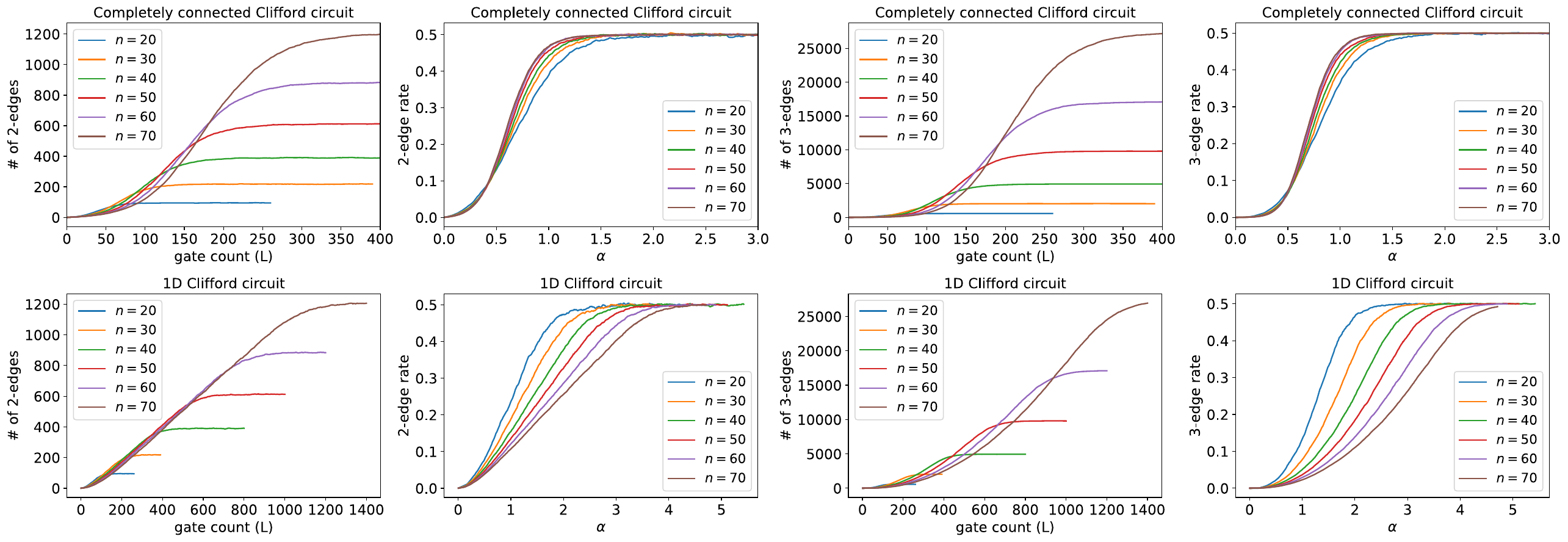}
    \caption{The number of hyperedges of the final frame for random Clifford circuits with completely-connected and 1D-connected architecture. At each gate count $L = \alpha n \ln n$, we averaged the number of hyperedges over 200 samples. The rate denotes the ratio between the number of hyperedges to the maximum possible numbers, which are $\binom{n}{2} = \frac{n(n-1)}{2}$ and $\binom{n}{3} = \frac{n(n-1)(n-2)}{3!}$ for 2 and 3-edges, respectively.}\label{sup:num_edges}    
\end{figure*}

\subsection{Hypergraph of the frame function by increasing gate counts}
We discuss how the frame function and its corresponding hypergraph behave when we increase the number of gates in a quantum circuit. We first note that by applying each Clifford gate, the update of the frame function is governed by two factors: the phase function ($P(\bfa)$) and the symplectic transformation ($S$). More precisely, following the arguments in Sec.~\ref{sup1:secB}, if the frame function after applying the $\ell$-th Clifford gate is given by $F_{\ell}(\bfa) = \sum_{\mu, \nu} c_{(\mu, \nu)} a_\mu a_\nu + \sum_{\mu, \nu, \omega} c_{(\mu, \nu, \omega)} a_\mu a_\nu  a_\omega$, after applying the $(\ell+1)$th Clifford gate, the frame function is updated to 
$$
F_{\ell+1}(\bfa) = F_{\ell}(S^{-1}\bfa)+P(S^{-1}\bfa) = F_{\ell}(\bfa) + \Delta F_\ell(\bfa),
$$
where we denote $\Delta F_\ell(\bfa)$ as the updated frame function. We can further decompose the updated frame function  into
$$
\Delta F_\ell(\bfa) = \Delta F^S_\ell (\bfa) + \Delta F^P_\ell (\bfa),
$$
where $\Delta F^S_\ell (\bfa)$ is the contribution from the syplectic transformation $S$ and $ \Delta F^P_\ell (\bfa)$ is the contribution from the phase function $P$. For example, by adding $CNOT_{i\rightarrow j}$ in Table~\ref{sup:table}, the transform $S^{-1} \bfa$ is given by
$$
[S^{-1} \bfa]_{\mu} = 
\begin{cases}
a_{iz} + a_{jz} & (\mu = {iz})\\
a_{ix} + a_{jx} & (\mu = {jx}) \\
a_\mu & ({\rm otherwise}),
\end{cases}
$$
which can be written as $[S^{-1}\bfa]_\mu = a_\mu + \delta_{\mu, iz} a_{jz} + \delta_{\mu, jx} a_{ix}$ in terms of the delta function, $\delta_{\mu,\mu'} = 0$ for $\mu \neq \mu'$ and $\delta_{\mu,\mu'} = 1$ for $\mu = \mu'$.
This leads to 
$$
\begin{aligned}
F_{\ell}(S^{-1}\bfa) &= \sum_{\mu, \nu} c_{(\mu, \nu)} [S^{-1} \bfa]_\mu [S^{-1} \bfa]_\nu + \sum_{\mu, \nu, \omega} c_{(\mu, \nu, \omega)} [S^{-1} \bfa]_\mu [S^{-1} \bfa]_\nu [S^{-1} \bfa]_\omega \\
&= \sum_{\mu, \nu} c_{(\mu, \nu)} (a_\mu + \delta_{\mu, iz} a_{jz} + \delta_{\mu, jx} a_{ix}) (a_\nu + \delta_{\nu, iz} a_{jz} + \delta_{\nu, jx} a_{ix})  \\
&\quad + \sum_{\mu, \nu, \omega} c_{(\mu, \nu, \omega)} (a_\mu + \delta_{\mu, iz} a_{jz} + \delta_{\nu, jx} a_{ix}) (a_\nu + \delta_{\nu, iz} a_{jz} + \delta_{\nu, jx} a_{ix}) (a_\omega + \delta_{\omega, iz} a_{jz} + \delta_{\omega, jx} a_{ix}) \\
&= F_\ell(\bfa) + \Delta F^S_\ell(\bfa),
\end{aligned}
$$
where $\Delta F^S_\ell (\bfa)$  adds additional terms to the frame function for each monomial of $F_\ell(\bfa)$ contains $a_{iz}$ or $a_{jx}$ (i.e., all the terms containing the delta functions). Meanwhile, $P(S^{-1}\bfa)$ adds additional terms in the frame function corresponding to $\Delta F^P_\ell(\bfa) = P(S^{-1}\bfa) = [S^{-1} \bfa]_{jz} [S^{-1} \bfa]_{ix} ([S^{-1} \bfa]_{iz} + [S^{-1} \bfa]_{jx}+1) = a_{jz} a_{ix} (a_{iz} + a_{jx} + 1)$. In contrast to $\Delta F^S_\ell(\bfa)$, we note that $\Delta F^P_\ell(\bfa)$ is independent of the previous frame function $F_\ell(\bfa)$. The updated frame functions can similarly be obtained for other Clifford gates based on Table~\ref{sup:table}.

In terms of the hypergraph, for each monomial in the updated frame function $\Delta F_\ell(\bfa)$, the update is given by the following rules:
\begin{enumerate}
\item The corresponding hyperedge is added when the hypergraph does not contain it.
\item The corresponding hyperedge is removed if the hypergraph already contains the updated hyperedge.
\end{enumerate}
As the hypergraph does not have any hyperedges at the beginning, at the earlier stage, the hypergraph is updated in a way that new hyperedges are added. However, as the number of hyperedges increases, there is a higher chance that the hypergraph already contains the hyperedge to be updated, in which case, the hyperedge is removed. Eventually, the number of newly generated hyperedges will balance the removed hypergraphs as the graph's hyperedges approaches its half maximum. As the hypergraph’s connectivity saturates at half-maximum (i.e., half-connected), the size of a minimum vertex cover also becomes large. Consequently, the number of efficiently simulatable marginal qubits becomes small as the vertex cover corresponds to the qubit that should be traced out.

The numerical simulation (Fig.~\ref{sup:num_edges}) shows that the numbers of both $2-$ and $3-$hyperedges saturate to the half maximum. For a completely connected Clifford circuit, the saturation rapidly happens at the gate count $L = \alpha n \ln (n)$, which corresponds to the circuit depth of ${\cal O}( (\ln n )^2)$. Furthermore, the phase transition observed in the completely connected circuit with gate count $L = \alpha n \ln n$ (Fig. 3(d) in the main text and related discussion in Sec.~\ref{sup:FSS}) indicates that the number of simulatable qubits is highly suppressed even within log-depth. Meanwhile, due to their local connectivity, 1D circuits undergo a moderate increase in the number of hyperedges toward the half maximum. We expect that this is closely related to the observation that the number of classically simulatable qubits linearly scales for log-depth 1D Clifford circuit (Fig. 3(b) in the main text) while leaving an open question of whether a non-vanishing scaling of simulatable qubits beyond log-depth can be found.

\subsection{Simulation time comparison with Qiskit}
\begin{figure*}[t]
\centering
    \centering
    \subfigure[]{
    \includegraphics[width=8cm]{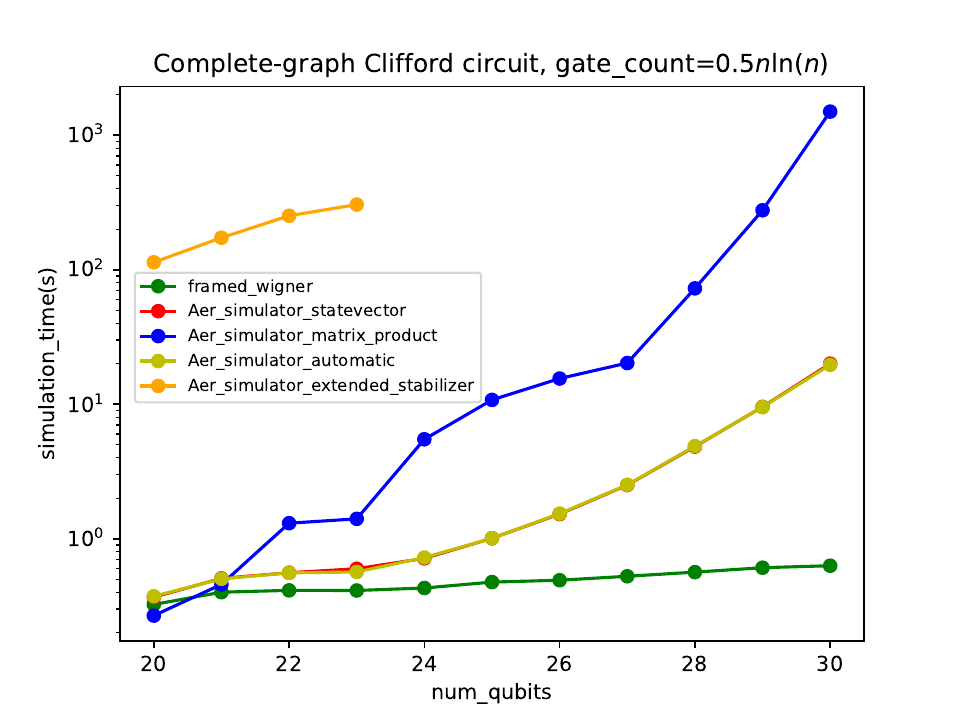}}
    \subfigure[]{
    \includegraphics[width=8cm]{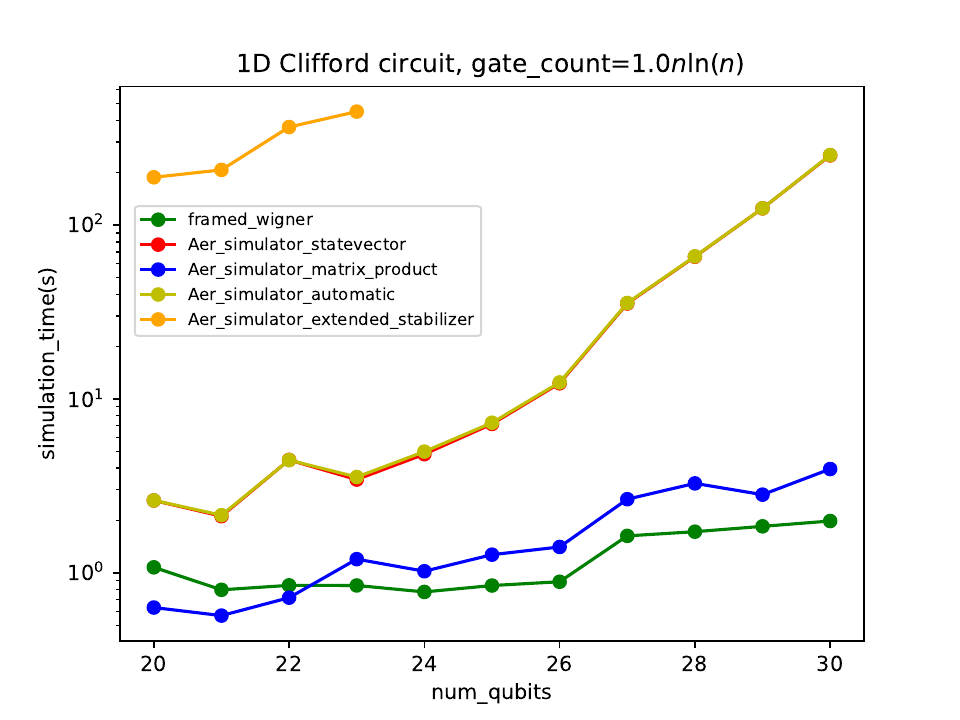}}
    \subfigure[]{
    \includegraphics[width=8cm]{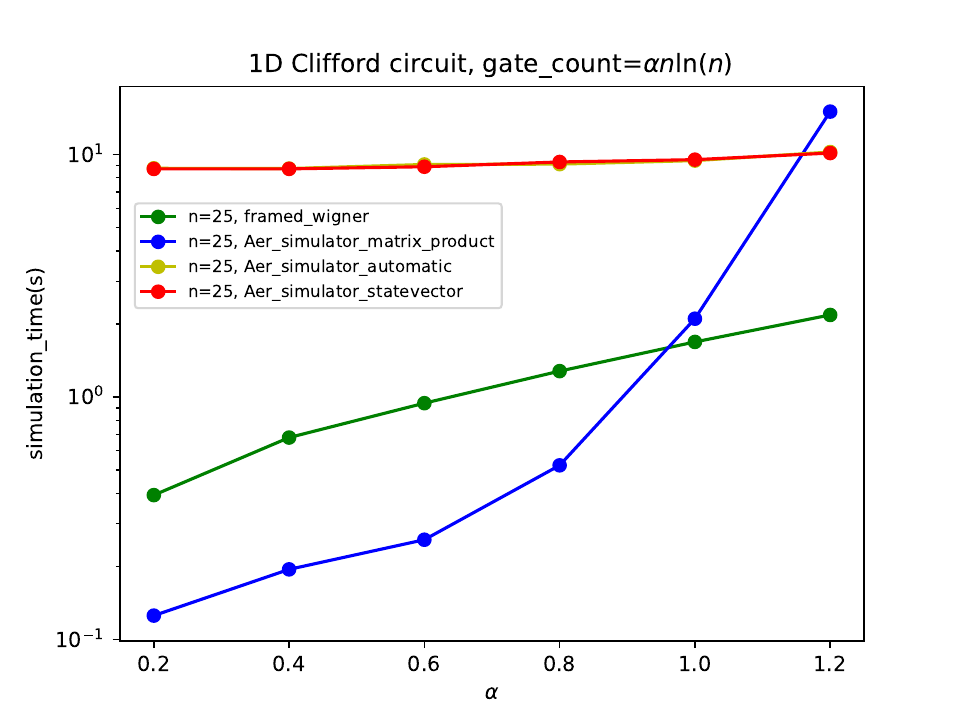}}
    \subfigure[]{
    \includegraphics[width=8cm]{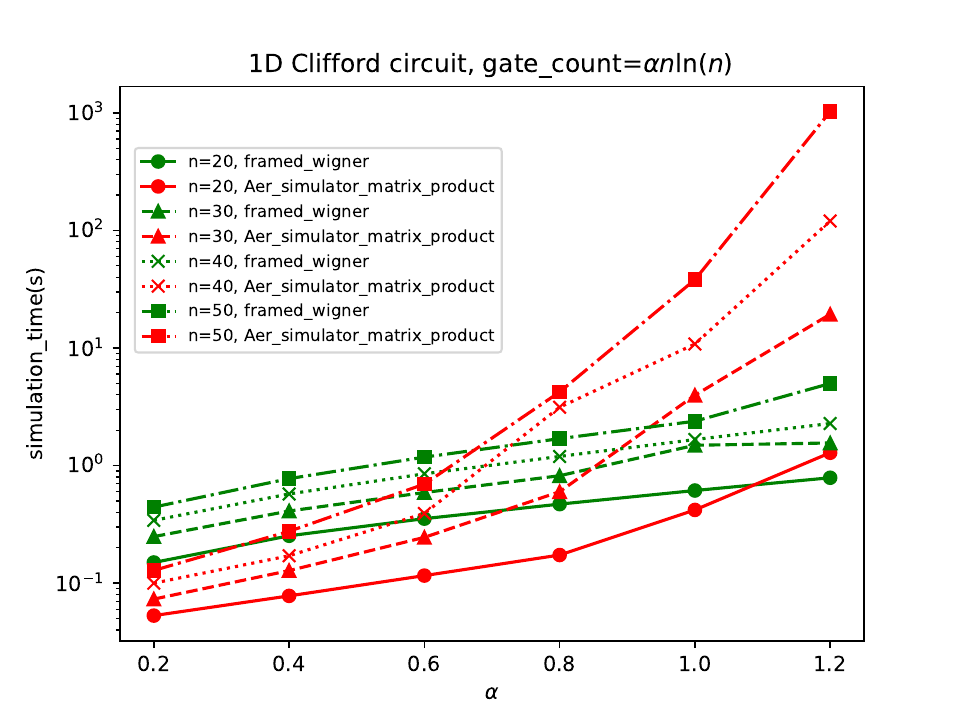}}
    
    \caption{Comparison of simulation time for the log-depth random Clifford circuits between framed Wigner function method and Qiskit Aer simulators. (a,b) Average time (over 200 random samples) for simulating the log-depth Clifford circuits. We fixed the gate count scaling factor $\alpha$, while increasing the number of qubits. (c,d) Average time (over 100 random samples) for simulating the 1D Clifford circuits by using the framed Wigner function method and Qiskit Aer simulators with varying factors of depth. We took the $n$-copies of $\ket{A}$ state as an input for both (a, b, c, d). Here, Aer\textunderscore stabilizer simulator is not functioning because given circuits have non-stabilizer state input, nor is Aer\textunderscore density\textunderscore matrix by the excessive memory issue.}\label{sup:fig3}    
\end{figure*}
In this section, we compare the simulation time of shallow non-adaptive Clifford circuits (see Fig.~\ref{sup:fig3}) between our simulator and the Qiskit Aer\textunderscore simulators.
For the Wigner function method, simulation time includes the frame changing and finding simulatable qubits, as well as sampling and rotating the phase point. In both simulators, Wigner and Qiskit, we measure the marginal outcome only once. Given that we have a circuit depth $d$, the total time complexity for frame changing is $\mathcal{O}(dn^4)$ because the gate count is $\mathcal{O}(dn)$. The time complexity $\mathcal{O}(\alpha\poly(n))$ in the main text is obtained by the gate count $\alpha n\ln n$. From Sec.~\ref{sup3:secA}, the time complexity for solving the vertex cover problem is $\mathcal{O}(n^4)$. Hence the total time complexity is $\mathcal{O}(\alpha n^4\ln(n))+\mathcal{O}(n^4)\le \mathcal{O}(\alpha \poly(n))$.

However, when $\alpha$ is low such that the circuit is not fully connected, the time complexity is lower than the above worst case. To be specific, we suppose that the given 1D circuit has a depth $d$. Then, two qubits with over $2d$-interval cannot be connected. This means that the degree of vertices in the resulting frame graph does not exceed over $\mathcal{O}(d^2)$, and the number of edges of the resulting frame graph is $\mathcal{O}(nd^2)$. Given that the gate count is $\mathcal{O}(dn)$, the time complexity for the frame changing is $\mathcal{O}(d^3 n^2)$. Furthermore, from Sec.~\ref{sup3:secA}, the time complexity for solving the vertex cover problem is $\mathcal{O}(n^2 d^2)$. In conclusion, given that $d=\alpha \mathcal{O}(\ln (n))$, the total time complexity of frame changing and vertex cover solving is $\mathcal{O}(\alpha^3 n^2 (\ln(n))^3)$.

In Fig.~\ref{sup:fig3} (a,b), we observe that our simulator executes the marginal sampling with the polynomial scaling of the time costs by increasing the number of qubits, which can be compared to the Aer\textunderscore simulators whose time complexity, except for Aer\textunderscore matrix\textunderscore product\textunderscore state (Aer\textunderscore mps), increases exponentially by increasing the number of qubits. From Fig.~\ref{sup:fig3} (c,d), the Aer\textunderscore mps performs better than framed Wigner when the circuit is low-entangled or has few qubits. Because the Aer\textunderscore mps employs the tensor network method and is efficient for circuits with large $n$ but with low-entanglement~\cite{vidal2003}. However, we note that for $n\ge 30$, the simulation time of Aer\textunderscore mps increases exponentially by increasing the scale factor of depth, $\alpha$. Hence, our method outperforms Aer\textunderscore mps for this region.

When the number of simulatable qubits is not sufficiently large, there could be other methods, for example, by directly calculating the Born probability to simulate Clifford circuits~\cite{aronson2004}. However, we expect that the low time scaling of the Wigner function simulator, while keeping a sufficient portion of simulatable qubits, leads to more efficient simulation for larger $n$ compared to the previous methods.



\subsection{Finite-sized scaling for log-depth completely connected circuits}
\label{sup:FSS}
Here, we explain the finite-sized scaling (FSS) analysis of Fig.~3 (e) in the main text. For the completely connected random Clifford circuits, the average number of measurable qubits ($n_{{\rm sim}}$) has two different scaling on $n$ for $\alpha > \alpha_c$ and $\alpha < \alpha_c$ with some critical point $\alpha_c$. In order to explore the critical point, we model the scaling function of $n_{{\rm sim}}$ in terms of $\alpha$ and $n$ as,
\begin{equation}
    n_{{\rm sim}}(\alpha,n)-n_{{\rm sim}}(\alpha_c,n)=f((\alpha-\alpha_c),n),
\end{equation}
for some function $f$ with two different scalings at $\alpha<\alpha_c$ and $\alpha>\alpha_c$. It naturally follows that $f(0,n)=0$ for all $n$ at the critical point. In order to estimate the critical values, we apply the FSS method to the data set with various $\alpha$ and $n$ values. We take the ansatz $f((\alpha-\alpha_c)n^{\frac{1}{\nu}})$ for some $\nu\in\mathbb{R}$ which is commonly found in the FSS literature~\cite{skinner2019}. By numerically optimizing the parameter of $\alpha_c$ and $\nu$ from the data set, we obtain a good collapse of data as in Fig. 3 (e) in the main text with the optimal parameter $\alpha_c = 0.81$ and $\nu = 1.28$. After we find the $(\alpha_c, \nu)$, we observe  $n_{{\rm sim}}(\alpha,n)-n_{{\rm sim}}(\alpha_c,n)\sim C(\alpha-\alpha_c)n^{\frac{1}{\nu}}\;(C\in \mathbb{R}^{+})$ when $\alpha < \alpha_c$. On the other hand, for $\alpha>\alpha_c$, $n_{{\rm sim}}(\alpha,n)$ shows almost flat behavior when increasing $n$ (see Fig. ~3(d) in the main text). 

\subsection{On the inseparability of efficiently measurable qubits found from the greedy algorithm for 1D circuits}\label{sup:sec4_1d}

For the $\log$-depth 1D architecture, even if we do not use the framed Wigner formalism, we can always find many qubits on which measurements are simulatable via brute-forced matrix calculation. We briefly explain how to do this and numerically show that simulatable qubits found by our method do not fall to this case. Suppose we have a 1D circuit with the depth $d=\mathcal{O}(\ln n)$. If two measurements are $2d$-far way from each other, then the first measurement does not influence the measurement outcome of the others (see Fig.~\ref{sup:fig1} (a) for the $d=3$ case). It means that we can sample the each outcome by calculating its Born probability with $\mathcal{O}(d^2)$-time. This is because during the calculation, we need to make Z-operators backward-evolved by $d$-depth 1D circuit and its time cost is proportional to the number of twirling gates. Therefore, since $\frac{n}{2d}=\Omega(\frac{n}{\ln n})$-number of qubits are separated each other, we can simulate those qubits within $\mathcal{O}(nd)=\mathcal{O}(n\ln n)$-time.

\begin{figure*}[t]
\centering
    \centering
    \subfigure[]{
    \includegraphics[width=5.5cm]{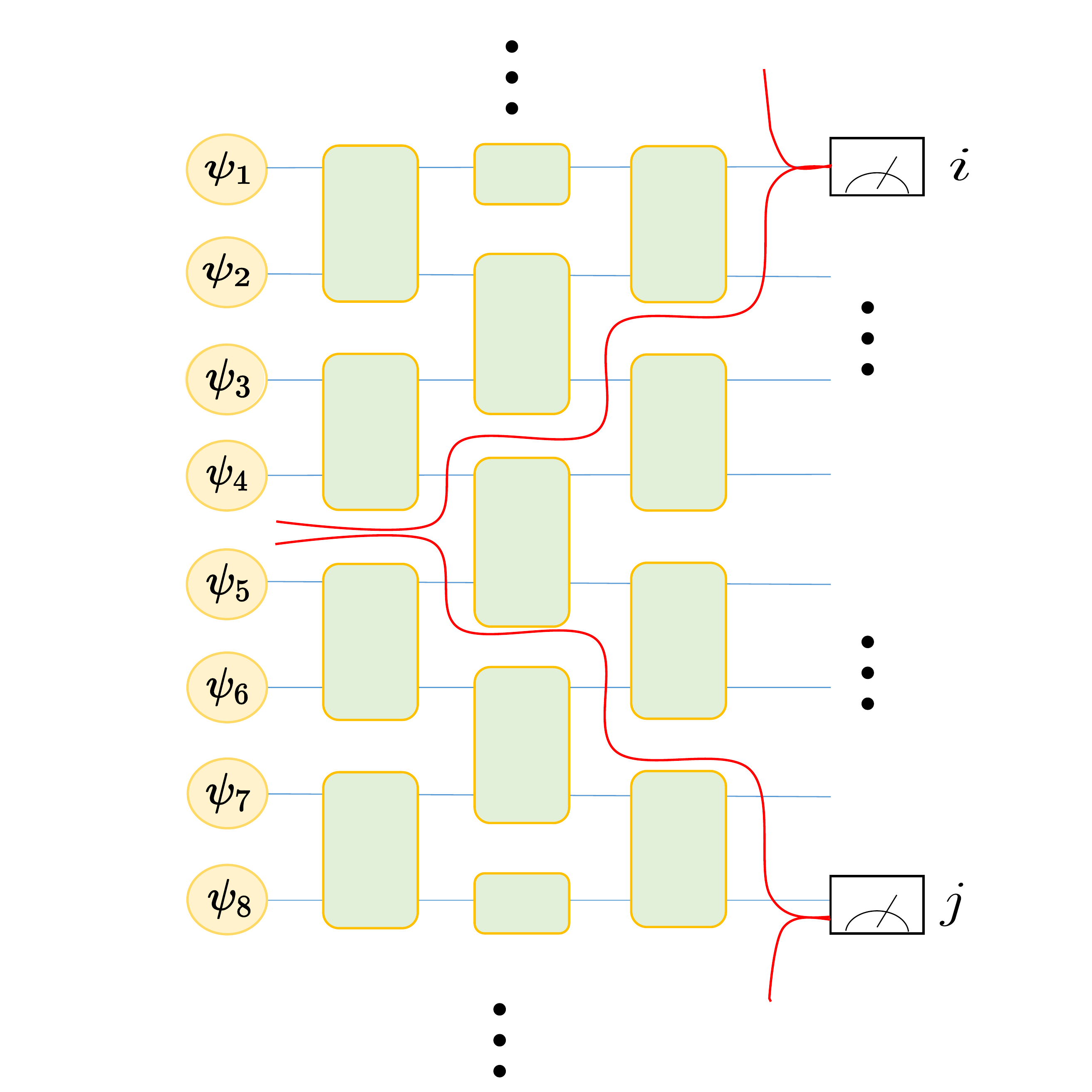}}
    \subfigure[]{
    \includegraphics[width=5.5cm,height=5cm]{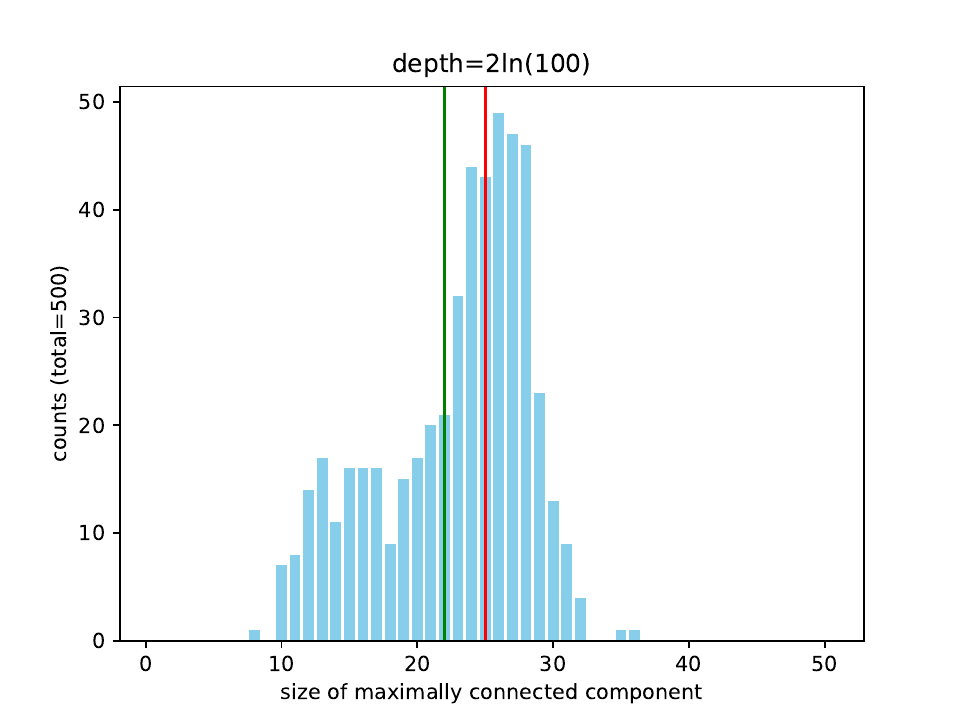}}
    \centering
    \subfigure[]{
    \includegraphics[width=5.5cm,height=5cm]{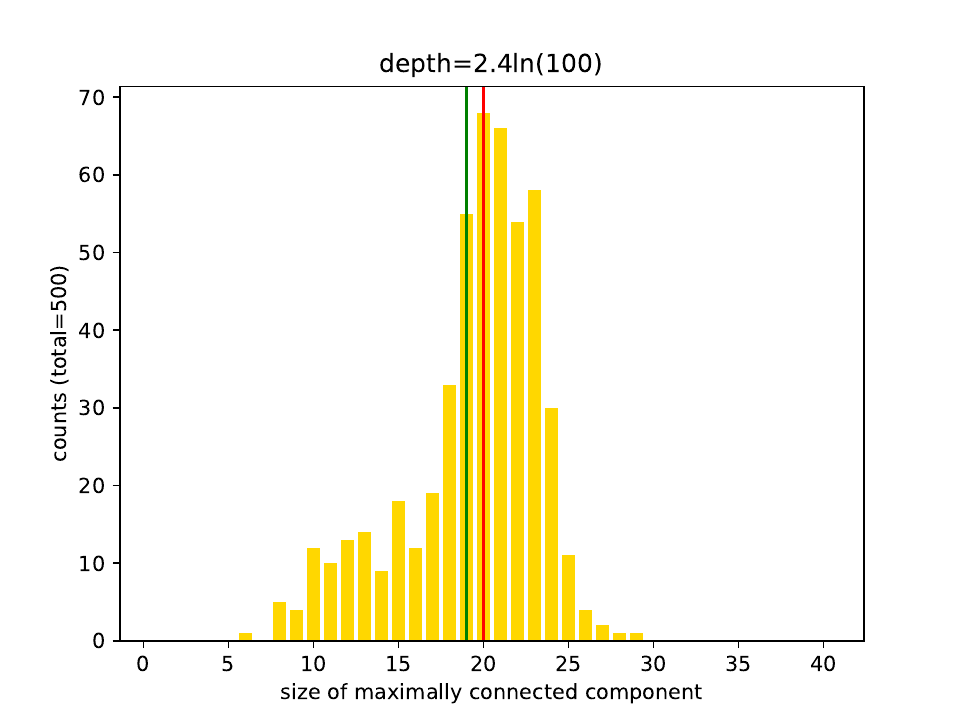}}
    \caption{ (a) Schematic illustration of a 3-depth 1D circuit. Measurement on the $i$-th qubit does not influence the measurement outcome of the $j$ (6-far away from $i$). If we only simulate these two measurements, we can separate this circuit into two portions sided by the red line. (b,c) The population of the size of maximally connected components for 500 numbers of 100-qubit 1D shallow Clifford samples. The green line represents the average size of maximally connected components, and the red line indicates the average value of total simulatable qubits of each sample. (b): Results of depth $=2\ln(100)$ 1D circuits (green line:22, red line:25).  (c): Results of depth $=2.4\ln(100)$ 1D circuits (green line:19, red line:20).}\label{sup:fig1}
\end{figure*}

However, this trivial method only finds distant groups with neighboring qubits.
In contrast, our methods also find simulatable qubits that are not far from each other, which can be numerically checked. More precisely, given a non-adaptive Clifford circuit with the Clifford unitary $U$, let $A_U\equiv \left\{i_1,i_2,\ldots, i_{\left|A_U\right|}\right\}\subset [n]$ be a set of locations of simulatable qubits. Now, let $Z_{i}$ be an $n$-qubit Paul operator, which affects the $Z$-operation to the $i$-th qubit and identity to the other qubits. We calculate the set of Pauli operators $P_{A_U}\equiv\left\{U^{\dagger}Z_{i}U|i\in A_U\right\}$ which can be obtained efficiently \cite{gottesman1998}. We then define the 2-graph $G(A_U,E_{A_U})$ where the set of vertex is $A_U$, and $E_{A_U}$ is a set of edges which connects $i,j\in A_{U}$ if and only if two Pauli strings $U^{\dagger}Z_{i}U$ and $U^{\dagger}Z_{j}U$ have non-trivial (not identity, not need to be same) Pauli operation on the same qubit locations. Now, we denote $\mathcal{C}(G(A_U,E_{A_U}))\subset A_U$ as a set of connected components of $G$, i.e., vertices consisting of connected subgraphs such that each subgraph is not connected with the others. We note that measurement on the set of qubits with location $C'\in\mathcal{C}(G(A_U,E_{A_U}))$ is independent of the measurement outcome on the set of other locations in $\mathcal{C}(G(A_U,E_{A_U}))$. Hence, the time complexity of weak simulation is upper bounded by $\mathcal{O}\left(\exp\left({\max_{C\in\mathcal{C}(G(A_U,E_{A_U}))}(\left|C\right|)}\right)\right)$, which can be achieved by the similar method with the first paragraph. Now we see Fig.~\ref{sup:fig1} (b) and (c). We randomly sample 500 numbers of $100$-qubit 1D Clifford circuit $U$ with zero frame input and find $\max_{C\in\mathcal{C}(G(A_U,E_{A_U}))}(\left|C\right|)$ efficiently via Python NetworkX packages, and we observed that most of the samples have many connected Pauli operators $U^{\dagger}Z_{i}U\;(i\in A_U)$ and hence such a trivial decomposition (in the first paragraph) is not applied to them. 

\subsection{Faster frame changing and solving vertex cover via adjacency list}\label{sup:fasterframe_shallow}
In this section, we introduce a faster way to update the frame and to solve the vertex cover problem by using the adjacency list, an alternative representation of a (hyper)graph. To this end, we define an extended form of the adjacency list by including the degree of each vertex. This enables us to proceed with frame changing and solving the vertex cover problem in $\mathcal{O}(dn^3)$-time, lower than the previous bound $\mathcal{O}(dn^4)$. We give the definition first.
\begin{definition}
    (i) Given a hypergraph $G(V,E)$, the \emph{adjacency list} $A_G$ is a $|V|\times \left(\frac{|V|^2-|V|}{2}+2\right)$ binary 2-dimensional list where rows represent each vertex $v$ and columns except for the last represents a subset $S$ of $V\backslash\left\{v\right\}$, where $|S| \le 2$ and its elements are as follows:
    \begin{align}
        (A_G)_{v,S}=\begin{cases}
            1&(\left\{v\right\}\cup S\in E)\\
            0&({\rm otherwise}).
        \end{cases}
    \end{align}
    Finally, elements of the last column indicate the degree of each vertex.\\
    (ii) Given a third-degree frame function $F$ defined in Eq.~\eqref{eq:frame_function}, we take the set $V=\left\{1x,2x,\ldots,nx,1z,2z,\ldots,nz\right\}$ and $E=\left\{\left\{i\right\}|c_{(i)}=1\right\}\cup\left\{\left\{i,j\right\}|c_{(i,j)}=1\right\}\cup\left\{\left\{i,j,k\right\}|c_{(i,j,k)}=1\right\}$. Then the \emph{adjacency list of frame function $F$} is $A_{G(V,E)}$.
\end{definition}

For example, consider a hypergraph $G(V,E)$ with $V=\left\{v_1,v_2,v_3\right\}$ and $E=\left\{\left\{1,2\right\},\left\{2,3\right\},\left\{1,2,3\right\}\right\}$. Then its adjacency list is expressed as, 
\begin{align}
    A_{G}=
    \left[
    \begin{array}{ccc|c}
        1&0&1&2\\
        1&1&1&3\\
        0&1&1&2
    \end{array}
    \right].
\end{align}
The row elements share the same vertex $v$ so that we can check the existence of given edge $e\in E$ by first choosing the row of $v\in e$ and finding the binary value corresponding to $S=e\backslash\left\{v\right\}$. This can be done in constant time. 

Now, we again apply the frame-changing algorithm for the Clifford gate acting on $(i,j)$-th qubits, shown in Sec.~\ref{supp:sec_complexity}. However, if we start from the adjacency list of the input frame, we can select the monomials in $\mathcal{O}(n^2)$-time by searching the $1$'s only in the rows of $\left\{ix,jx,iz,jz\right\}$. After the symplectic transformation, with $P(S^{-1}(\bfa))$, the number of generated monomials is $\mathcal{O}(n^2)$ and summation to the initial frame takes $\mathcal{O}(n^2)$-time. This is because each generated monomial has at most $\mathcal{O}(1)$-number of rows, and we update the graph by flipping a corresponding location of a binary element in the adjacency list and by adjusting the degree of the vertex at the last column. In conclusion, the total time complexity of frame changing after a single Clifford gate takes $\mathcal{O}(n^2)$-time, which is lower than the time cost $\mathcal{O}(n^3)$ described in Sec.~\ref{supp:sec_complexity}. Given that the gate count is $\mathcal{O}(dn)$, total time complexity becomes $\mathcal{O}(dn^3)$.

Next, to obtain classically simulatable marginal qubits, let us consider the adjacency list of the final frame $F({\bf 0}_x, \bfa_z)$, where taking $\bfa_x = {\bf 0}_x$ is equivalent to removing all the vertices $\{ ix \}_{i=1}^n$, which takes $\mathcal{O}(n^3)$-time. We then proceed with the greedy algorithm to solve the vertex cover problem for the reduced graph with $V=\left\{1z,2z,\ldots,nz\right\}$. From the last column, we can find the vertex with the largest degree $v'$ in $\mathcal{O}(n)$-time. Eliminating each edge containing $v'$ from the adjacency list takes $\mathcal{O}(1)$-time because this edge can be found only once in $\mathcal{O}(1)$-number of rows. For each time the element is changed to zero, we should subtract the rightmost value of the given row by $1$. Therefore, setting all elements, corresponding to $v'$-th row or column $v'\in S$, to $0$ takes $\mathcal{O}(d(G))$-time. By repeating this procedure until all the elements of the last column become zero, which takes at most as many steps as the number of vertices $|V| = n$, we find the vertex cover. Hence, given the final frame $F(\mathbf{0}_x,\bfa_z)$, we note that the time complexity for the greedy algorithm to solve the vertex cover becomes $\mathcal{O}( n (n + d(G_{F}))) = \mathcal{O}(n^2+nd(G_{F}))$.

\section{More general frame formalism and Born probability estimation} \label{sup:secB}

\subsection{Wigner representation using multiple frames} \label{supsec6:A}

In this section, we further generalize the framed Wigner function formalism by taking multiple frame functions to represent a quantum state. In particular, we show that an arbitrary product state can be positively represented by using multiple quadratic frame functions.


We recall that a single qubit state can be expressed by either zero frame ($F=0$) or conjugate frame ($F=a_{1x}a_{1z})$. By combining the phase space operators corresponding to these two frames, one can express any single qubit states by the convex sum of these operators. Hence, if $\rho$ is a product state, it is positively represented under a set of frames $\mathcal{F}\equiv \left\{\bfb\cdot (a_{1x}a_{1z},a_{2x}a_{2z},\ldots,a_{nx}a_{nz})|\bfb\in\bfZ^n_2\right\}$.
    More precisely, for $\rho=\bigotimes_{i=1}^{n}\rho_i$, we express each $\rho_i$ as 
    \begin{align}
        \rho_i=\sum_{\bfu_i\in\bfZ^2_2}\left(w^0_{\rho_i}(\bfu_i)A^0(\bfu_i)+w^{a_{ix}a_{iz}}_{\rho_i}(\bfu_i)A^{a_{ix}a_{iz}}(\bfu_i)\right),
    \end{align}
    using eight phase space operators $\{ A^0(\bfu_i), A^{a_{ix} a_{iz}}(\bfu_i)\}$, where $\bfu_i\equiv(u_{ix},u_{iz}) \in \bfZ_2^2$. Here, $w^0$ and $w^{a_{ix}a_{iz}}$ are non-negative functions and are different from the Wigner functions $W_\rho^F(\bfu)$ in the main text. We note that summation over phase point of $w^0_{\rho_i}(\bfu_i)$, or $w^{a_{ix}a_{iz}}_{\rho_i}(\bfu_i)$ solely does not give unity, but $\sum_{\bfu_i}\left(w^0_{\rho_i}(\bfu_i)+w^{a_{ix}a_{iz}}_{\rho_i}(\bfu_i)\right)=1$. We also note that these two functions cannot be obtained by inversion in Eq. (2) of the main text because we now have eight phase point operators, which are overcomplete, and hence, their coefficients are not unique. However, we can still efficiently find these two functions because the convex polytope by 8-phase point operators as extreme points contain the Bloch sphere (see Ref.~\cite{raussendorf2017}). Hence, solving the constant-sized system of linear equations leads to the following expression of $\rho$,
    \begin{align}\label{sup:eqborn1}
        \rho&=\bigotimes_{i=1}^{n}\rho_i=\sum_{\bfb\in\bfZ^n_2}\sum_{\bfu\in V_{n}}\left(\prod_{i=1}^{n}w^{b_ia_{ix}a_{iz}}_{\rho_i}(\bfu_i)\right)\left(\bigotimes_{i=1}^{n}A^{b_ia_{ix}a_{iz}}(\bfu_i)\right),
    \end{align}
    where $\bfu=\bigoplus_{i=1}^{n}(\bfu_i)$. From the definition, $\bigotimes_{i=1}^{n}A^{b_ia_{ix}a_{iz}}(\bfu_i)$ is an $n$-qubit phase point operators with a frame function $F=\bfb\cdot (a_{1x}a_{1z},a_{2x}a_{2z},\ldots,a_{nx}a_{nz})$. Hence, we can rewrite Eq.~\eqref{sup:eqborn1} as
    \begin{align}
    \rho=\sum_{F\in\mathcal{F}}\sum_{\bfu\in V_{n}}\left(\prod_{i=1}^{n}w^{b_ia_{ix}a_{iz}}_{\rho_i}(\bfu_i)\right)(A^{F}(\bfu)).
    \end{align}
    Therefore, the desired Wigner function is 
    \begin{align}\label{sup:eqborn2}
        W^F_{\rho}(\bfu)=\prod_{i=1}^{n}w^{b_ia_{ix}a_{iz}}_i(\bfu_i).
    \end{align}
Although there could not be a unique expression, the same phase point sampling protocol can be applied to the generalized Wigner functions with multiple frames whenever they are non-negative. This is because the non-negative function $W^F_{\rho}$ is a probability distribution with random variables of not only $\bfu$ but also $F$, so that we can sample both $\bfu$ and $F$ from $W^{F}_{\rho}(\bfu)$. For this case, we say $\rho$ is positively represented under a \emph{frame set} $\mathcal{F}$. Since the Wigner function has a product form and for each $i\in[n]$, we sample $\bfu_i$ and $b_i$ from $W^{b_ia_{ix}a_{iz}}_i(\bfu_i)$. The resulting sampling outcome then becomes $\bfu=\bigoplus_{i}^{n}(\bfu_i)$ and $F=\bfb\cdot (a_{1x}a_{1z},a_{2x}a_{2z},\ldots,a_{nx}a_{nz})$.

\subsection{Generalization of weak simulation results}

With the above generalization of using multiple frames representation, we can also obtain a more general statement on the weak simulation. We will discuss this in detail in the next theorem.

\begin{theorem}\label{sup:thm1}
Suppose that we have a product state $\rho=\bigotimes_{i=1}^{n}\rho_i$ as an input and a non-adaptive Clifford operation $U$ and $Z$-measurements. Then $\rho$ is positively represented under a frame set $\mathcal{F}\equiv \left\{\bfb\cdot (a_{1x}a_{1z},a_{2x}a_{2z},\ldots,a_{nx}a_{nz})|\bfb\in\bfZ^n_2\right\}$. We denote a resulting frame function via the Clifford circuit starting from $F_{\rm in}\in \mathcal{F}$ as $F$. Also, let the $I\subset[n]$ satisfies that for all $F_{\rm in}\in \mathcal{F}$, $F(\bfa_x=\mathbf{0},\bfa_z)\big|_{a_{iz}=0\;{\rm for}\;i\notin I_F}$ is zero. We then can classically simulate $\left|I\right|$-number of measurements.
\end{theorem}

\begin{proof}
    We start from Eq.~\eqref{sup:change}. Probability to measure $\bfx\in\bfZ^n_2$ is

    \begin{align}
        {\rm Tr}(\ket{\bfx}\bra{\bfx}U\rho U^{\dagger})&=\sum_{\bfu,F_{\rm in}\in\mathcal{F},\bfa_z\in\bfZ^n_2}\frac{1}{2^n}W^{F_{\rm in}}_{\rho} (S^{-1}(\bfu)) \left(-1\right)^{(\bfu+\bfx)\cdot\bfa_z+F_{\rm in}(S^{-1}(\mathbf{0}_x,\bfa_z))+P(S^{-1}(\mathbf{0}_x,\bfa_z))} \\
        &=\sum_{\bfu,F_{\rm in}\in\mathcal{F},\bfa_z\in\bfZ^n_2}\frac{1}{2^n}W^{F_{\rm in}}_{\rho} (S^{-1}(\bfu)) \left(-1\right)^{((\bfu+\bfx)\cdot\bfa_z+F(\mathbf{0}_x,\bfa_z))}.
    \end{align}
    Now, consider measuring a subset $I$ of qubits. The probability of measuring a marginal string $\bfx'$ (with arbitrary qubit locations) is as follows. There exists a set of $\bfv_{F}\in V_n\;(F_{{\rm in}}\in \mathcal{F})$ such that

    \begin{align}
        p(\bfx')&=\sum_{\bfx''}{\rm Tr}(\ket{\bfx'\oplus\bfx''}\bra{\bfx'\oplus\bfx''}U\rho U^{\dagger}) \\
        &=\sum_{\bfu,F_{\rm in}\in\mathcal{F}}\left(\sum_{\substack{\bfa_z\in\bfZ^n_2\\a_{iz}=0\;{\rm for}\;i\in\left([n]\backslash I\right)}}\frac{1}{2^{\left|I\right|}}W^{F_{\rm in}}_{\rho} (S^{-1}(\bfu+\bfv_{F})) \left(-1\right)^{((\bfu_x+(\bfx'\oplus 0))\cdot\bfa_z)}\right) \\
        &=\sum_{\bfu,F_{\rm}\in\mathcal{F}}W^{F_{\rm in}}_{\rho} (S^{-1}((\bfu+\bfv_{F})))\prod_{i\in I}\delta_{x_i u_{ix}}.
    \end{align}
    Therefore, we obtain the following simulation scheme given that the conditions in Theorem~\ref{sup:thm1} for the final frame hold.
    \begin{enumerate}
        \item Sample a phase space point $\bfu \in V_n$ and $F_{\rm in}\in\mathcal{F}$ from $W^{F_{\rm in}}_\rho(\bfu)$.

        \item Change the input frame under the given Clifford operation and obtain both final frame $F$ and $\bfv_{F}$
        
        \item Update $\bfu$ $\leftarrow$ $\bfu'\equiv S(\bfu)$.

        \item Update $\bfu\leftarrow \bfu+\bfv_{F}$
        \item Desired outcome is a marginal string $\bfu'_x$.
    \end{enumerate}
\end{proof}

We can see that by Theorem~\ref{sup:thm1} (i), the larger the frame set for quantum state input we represent, the fewer the number of measurable qubits. For example, $n$-copies of \emph{equatorial state}, $\ket{E_{\phi}}\equiv\frac{1}{\sqrt{2}}\left(\ket{0}+e^{i \phi}\ket{1}\right)$, can have non-negative representation by $2^n$-numbers of frame (see Fig.~4 of Ref.~\cite{raussendorf2017} for single qubit case). Whereas, the $n$-copies of $\ket{A}= \cos (\theta/2) \ket{0} + e^{i (\pi/4)}\sin(\theta/2)\ket{1}$ with $\theta = \cos^{-1}(1/\sqrt{3})$~\cite{qassim2021} need only zero frame for non-negative representation. Interestingly, for approximate simulation, there exists an algorithm that marginally simulates a large fraction of circuits with $\ket{E_{\frac{\pi}{4}}}^{\otimes n}$ as an input efficiently but not for $\ket{A}$ state input \cite{bu2019}, which has Pauli rank $4$.  

\subsection{Born probability estimation}\label{supsec6:C}

We consider the Born probability estimation of a quantum circuit with outcome $\bfx \in \bfZ^n_2$ within additive error $\epsilon$. Suppose a given quantum state $\rho$ is positively represented under the frame $F_{{\rm in}}$. From Eq.~\eqref{sup:linear_ig}, we have
\begin{align}\label{supp:born}
    {\rm Tr}(U\rho U^\dagger \ket{\bfx}\bra{\bfx})
    =\frac{1}{2^n}\sum_{\bfu\in V_n,F\in \mathcal{F}}\sum_{\bfa_z\in\bfZ^n_2}W^{F_{\rm in}}_{\rho}(S^{-1}(\bfu+\bfv_F))(-1)^{(\bfu_x+\bfx)\cdot \bfa_z+F(\mathbf{0}_x,\bfa_z)},
\end{align}
for some $\bfv_F\in V_n$. The first method is to uniformly choose $\bfa_z$ and take an estimator,
\begin{align}
    \hat{p}(\bfx)=\sum_{\bfu\in V_n,F\in \mathcal{F}}W^{F_{\rm in}}_{\rho}(\bfu)(-1)^{(S(\bfu)_x+(\bfv_F)_x+\bfx)\cdot \bfa_z+F(\mathbf{0}_x,\bfa_z)}=\sum_{\bfa_z\in\bfZ^n_2}(-1)^{\bfa_z\cdot\bfx}\left({\rm Tr}\left(U\rho U^{\dagger}T_{(\mathbf{0},\bfa_z)}\right)\right).
\end{align}
This estimator can be simulated classically if $W^{F_{\rm in}}_{\rho}(\bfu)$ has a product form. However, we also have another expression,  
\begin{align}\label{sup:stabest}
    {\rm Tr}(U\rho U^\dagger \ket{\bfx}\bra{\bfx})
    &=\frac{1}{2^n}\sum_{\bfa_z\in\bfZ^n_2}(-1)^{\bfa_z\cdot\bfx}\left({\rm Tr}\left(U\rho U^{\dagger}T_{(\mathbf{0},\bfa_z)}\right)\right)=\frac{1}{2^n}\sum_{\bfa_z\in\bfZ^n_2}(-1)^{\bfa_z\cdot\bfx}\left({\rm Tr}\left(\rho U^{\dagger}T_{(\mathbf{0},\bfa_z)}U\right)\right).
\end{align}

Hence, we may just uniformly randomly choose $\bfa_z$ and find $T'_{(\mathbf{0}_x,\bfa_z)}\equiv U^{\dagger}T_{(\mathbf{0},\bfa_z)}U$ by using the  stabilizer tableau \cite{gottesman1998}, and then take an estimator $\hat{p}_s(\bfx)=(-1)^{\bfa_z\cdot\bfx}{\rm Tr}(T'_{(\mathbf{0}_x,\bfa_z)}\rho)$. Therefore, the estimators $\hat{p}(\bfx)$ and $\hat{p}_s(\bfx)$ have same value for given sampled variable $\bfa_z$. Hence, both estimators have the same mean squared error, which is, 
\begin{align}
    {\rm Var}_{{\rm Pauli}}(\bfx)&=\frac{1}{2^n}\sum_{\bfa_z\in\bfZ^n_2}\left({\rm Tr}\left(\rho U^{\dagger}T_{(\mathbf{0},\bfa_z)}U\right)\right)^2-p(\bfx)^2 \\
    &=\frac{1}{4^n}\sum_{\bfx\in\bfZ^n_2}\sum_{\bfa_z,\bfb_z\in\bfZ^n_2}(-1)^{(\bfa_z+\bfb_z)\cdot\bfx}\left({\rm Tr}\left(\rho U^{\dagger}T_{(\mathbf{0},\bfa_z)}U\right)\right)\left({\rm Tr}\left(\rho U^{\dagger}T_{(\mathbf{0},\bfb_z)}U\right)\right)-p(\bfx)^2. \\
    &=\left(\sum_{\bfx\in\bfZ^n_2}p(\bfx)^2\right)-p(\bfx)^2=Z_{U,\rho}-p(\bfx)^2,
\end{align}
where $Z_{U,\rho}\equiv\sum_{\bfx\in\bfZ^n_2}p(\bfx)^2$ is the so-called collision probability \cite{dalzell2022}.

The above schemes do not use the probabilistic property of the Wigner function. We introduce another estimation method by sampling phase points from the Wigner function. We enclose it in the following result.

\begin{theorem}\label{sup:thm2}
    Assume $\rho=\bigotimes_{i=1}^{n}\rho_i$ is a product state, which is positively represented under a frame set $\mathcal{F}\equiv \left\{\bfb\cdot (a_{1x}a_{1z},a_{2x}a_{2z},\ldots,a_{nx}a_{nz})|\bfb\in\bfZ^n_2\right\}$. Also, let $F_0$ be the resulting frame via a given circuit starting from a zero frame. Now, we assume that $I_{F_0}\subset [n]$ satisfies that $F_0(\bfa_x=\mathbf{0},\bfa_z)\big|_{a_{iz}=0\;{\rm for}\;i\notin I_{F_0}}$ is a quadratic polynomial, or has $\mathcal{O}(\ln(n))$-sized vertex cover of subgraph of $G_{F_0}$ having only third-degree terms. Then there exists an efficient algorithm for estimation of marginal measurement probability $p(\bfx')$, where $\bfx'$ is a target string on qubits located on $I_{F'_0}$. Also, averaged variance over uniform outcomes is $(1-Z^{(k)}_{U,\rho})/2^k$. Hence if we uniformly randomly choose a binary string $\bfx'$,  ${\rm Var}(\bfx')\le\frac{a(k)}{2^k}-p(\bfx')^2$ with probability at least $\left(1-\frac{1}{a(k)}\right)$, where $a(k)$ is a non-negative function of $k$.
\end{theorem}

\begin{proof}
    We first consider the Born probability estimation of full string $\bfx\in\bfZ^n_2$. In the same manner as Eq.\eqref{sup:linear_ig}, we obtain that 
    \begin{align}
    &{\rm Tr}(U\rho U^\dagger \ket{\bfx}\bra{\bfx})=\frac{1}{2^n}\sum_{\bfu\in V_{n},F_{\rm in}\in\mathcal{F}}\sum_{\bfa_z\in\bfZ^n_2}W^{F_{\rm in}}_{\rho}(S^{-1}(\bfu+\bfv_{F}))(-1)^{(\bfu_x+\bfx)\cdot \bfa_z+F(\mathbf{0}_x,\bfa_z)}.
    \end{align}

    Now, we set the Born probability estimation algorithm of $P(\bfx)\equiv {\rm Tr}(U\rho U^\dagger \ket{\bfx}\bra{\bfx})$. Here's the scheme.
    \begin{enumerate}
    
    \item Sample a phase space point $\bfu \in V_n$ and $F_{\rm in}\in\mathcal{F}$ from $W^{F_{\rm in}}_\rho(\bfu)$. 

    \item Change the input frame under the given Clifford operation and obtain both final frame $F$ and $\bfv_{F_{\rm in}}$

    \item Update $\bfu\leftarrow S(\bfu)$ and then update $\bfu\leftarrow \bfu+\bfv_{F}$
    \item Desired estimation value for each trial is $\hat{p}(\bfx)\equiv\frac{1}{2^n}\sum_{\bfa_z}(-1)^{(\bfu_x+\bfx)\cdot \bfa_z+F(\mathbf{0}_x,\bfa_z)}$, where $F$ is a 
    final frame starting from $F_{\rm in}$. 
    \item Repeat step 2 $\sim$ step 4 to obtain many $\hat{p}(\bfx)$'s. The final estimation will be the sample mean of those $\hat{p}(\bfx)$'s.
    \end{enumerate}

    Unfortunately, this is not an efficient algorithm. Because at the third stage, $F(\mathbf{0}_x,\bfa_z)$ is in general of third-order. The exact calculation is \#P-Hard problem \cite{bremner2016}. However, in the cases where the size of the vertex cover of a hypergraph with terms of third-degree in $F(\mathbf{0}_x,\bfa_z)$ is $\mathcal{O}(\ln(n))$, we can do this efficiently \cite{montanaro2017}. Therefore, if we do not estimate the probability of all measurements, we can trace some qubits until the resulting frame is such a form. This is a more relaxed condition than one of weak simulation in Section~\ref{supsec1:C}. 
    If the final frames (after tracing) become quadratic, then $\hat{p}(\bfx)$ becomes an exponential sum of quadratic binary polynomials, which is efficiently calculated in $\mathcal{O}(k^3)$-time \cite{bravyi2019}.
    
    For every initial second-ordered frame we sampled from, all resulting frames (after we take $\bfa_x=0$) have the same third-ordered terms. Because third-ordered terms are obtained only from phase functions of Clifford gates and symplectic transforms, which do not raise the order of the polynomial, therefore the vertex cover problem may be solved only once for the resulting frame starting from $F_{\rm in}=0$.

    Now, we only see the marginal outcome string $\bfx'$. The mean squared error ${\rm Var}(\bfx')$ is given as, 
    \begin{align}\label{sup:bornupper}
    {\rm Var}(\bfx')
    =\frac{1}{4^k}\sum_{\bfu\in V_{n},F_{\rm in}\in\mathcal{F}}\sum_{\bfa'_z,\bfb'_z}W_{\rho}^{F_{\rm in}}(S^{-1}(\bfu+\bfv_F))(-1)^{(\bfu'_x+\bfx')\cdot(\bfa'_z+\bfb'_z)+F(\mathbf{0}_x,\bfa'_z\oplus 0)+F(\mathbf{0}_x,\bfb'_z\oplus 0)}-p(\bfx')^2,
    \end{align}
    where $\sum_{\bfa'_z,\bfb'_z}$ is the sum over strings at which the same marginalization is applied as $\bfx'$. Let us denote the first term of the right side as $\mathbb{E}(\hat{p}(\bfx')^2)$. When we take the uniform average to this over binary strings $\bfx'$, 
    \begin{align}
    \overline{\mathbb{E}(\hat{p}(\bfx')^2)}^{\bfx'}=\frac{1}{2^k}\sum_{\bfu,F_{\rm in}}W_{\rho}^{F_{\rm in}}(S^{-1}(\bfu+\bfv_F))
    =\frac{1}{2^k}. 
    \end{align}
    (Note that $\overline{{\rm Var}(\bfx')}^{\bfx'}=1/2^k-Z^{(k)}_{U,\rho}/2^k$, where $Z^{(k)}_{U,\rho}\equiv \sum_{\bfx'\in\bfZ^k_2}p(\bfx')^2$.) Hence, by Markov's inequality, when we uniformly randomly sample $\bfx'$, the probability of $\mathbb{E}(\hat{p}(\bfx')^2)$ being larger than $\frac{a(k)}{2^k}$ is,
    \begin{align}
    {\rm Pr}\left(\mathbb{E}(\hat{p}(\bfx')^2)\ge\frac{a(k)}{2^k}\right)\le\frac{\overline{2^k\mathbb{E}(\hat{p}(\bfx')^2)}^{\bfx'}}{a(k)}=\frac{1}{a(k)}.
    \end{align}
    We note that ${\rm Var}(\bfx')=\mathbb{E}(\hat{p}(\bfx')^2)-p(\bfx')^2$. Hence if we uniformly randomly choose $\bfx'$, the probability of ${\rm Var}(\bfx')\le \frac{a(k)}{2^k}-p(\bfx')^2$ is at least $\left(1-\frac{1}{a(k)}\right)$. 
\end{proof}

Suppose that we use product state input. We need to find the final (and marginal) frame from each initial frame, which can be sampled in $\mathcal{O}(n)$-time. Moreover, all $\mathcal{O}(k^2)$ number of coefficients of second-ordered terms in a final (and marginal) frame can be rewritten by a boolean linear function with the argument $\bfb\in\bfZ^n_2$ which represents the input frame sample $F_{\rm in}=\sum_{\bfb\in\bfZ^n_2}b_ia_{ix}a_{iz}$. These functions can be found before the sampling. 
Therefore, the total time for each trial is $\mathcal{O}(nk^2)$.

Whereas, we can easily derive that from Eq.~\eqref{sup:stabest}, ${\rm Var}_{{\rm Pauli}}(\bfx')=Z^{(k)}_{U,\rho}-p(\bfx')^2$ and calculation of $\hat{p}_s(\bfx')$ takes $\mathcal{O}(n)$-time given that $T'_{(\mathbf{0}_x,\bfa'_z)}\equiv U^{\dagger}T_{(\mathbf{0},\bfa'_z)}U$ is known.

The mean squared error is connected to the required number of samples to achieve additive estimation error $\epsilon$ with probability larger than $1-\delta$, given by $\mathcal{O}\left(\frac{{\rm Var}(\bfx)}{\epsilon^2}\ln(\frac{1}{\delta})\right)$ \cite{jerrum1986,blair1985}. Suppose that the input state is a product state. Using Theorem~\ref{sup:thm2}, we note that for any choice of a non-negative function $a(k) \geq 1$ at least $2^k(1-\frac{1}{a(k)})$-number of binary strings $\bfx'$ can be estimated with $\mathcal{O}\left(\frac{a(k)/ 2^k-p(\bfx')^2}{\epsilon^2}\ln(\frac{1}{\delta})\right)$ samples using the Wigner function approach. In contrast, Ref.~\cite{pashayan2020} requires $\mathcal{O}\left(\frac{Z_{U,\rho}^{(k)}-p(\bfx')^2}{\epsilon^2}\ln(\frac{1}{\delta})\right)$ samples for any string $\bfx'$. The total time to simulate is the product of the sample number and the time taken to run the estimator once. From the above arguments, if the input state is a product (by ignoring the time for the first trial) we have a time improvement for $2^k(1-\frac{1}{a(k)})$-number of strings when the collision probability $Z_{U,\rho}^{(k)}\ge\mathcal{O}(\frac{a(k)k^2}{2^k})$. The one of such cases is when $k\sim \beta n\; (\beta\in(0,1),\;n\;{\rm is}\;{\rm large})$ and the $k$-marginal probability distribution is far from \emph{anti-concentration}, $Z_{U,\rho}^{(k)}\sim \frac{1}{2^{\alpha k}}\; (\alpha\in (0,1))$ in which we take improvements for $2^k(1-\frac{1}{\poly(k)})$ number of strings.

\bibliographystyle{apsrev4-1}
\bibliography{ref}

\end{document}